\documentclass[journal]{IEEEtran}

\usepackage{pifont,amssymb,url}
\usepackage{multicol,mathrsfs}
\usepackage[usenames,dvipsnames,svgnames,table]{xcolor}
\usepackage{bm,bbm,mathrsfs,amssymb,pifont,amssymb,pifont,amsfonts,amsmath,stmaryrd,color,amssymb,graphics,graphicx,epstopdf}
\usepackage{amscd,graphicx,array,dsfont,texdraw,tikz,enumerate,multicol,multirow,verbatim}
\usetikzlibrary{arrows,shapes,shadows,positioning,automata,patterns}
\usetikzlibrary{matrix,arrows,calc,trees,decorations.pathmorphing,decorations.markings, decorations.pathreplacing}
\usepackage[utf8x]{inputenc}
\usepackage[T1]{fontenc}
\usepackage{booktabs}
\usepackage{algorithm,algorithmicx,algpseudocode} 

\usepackage{bm,amsmath,amssymb,amsthm,graphicx,subfigure,epstopdf}
\usepackage[comma,numbers,square,sort&compress]{natbib}
\usepackage{epstopdf,epic,subfigure,epsfig,arydshln,enumerate}

\usepackage{fancyhdr,lastpage,color,dsfont}

\newcommand{\myr}{\color{black}}
\definecolor{mitr}{RGB}{117,0,20}
\definecolor{mitSG}{RGB}{139,149,158}

\newtheorem{thm}{Theorem}
\newtheorem{cor}{Corollary}
\newtheorem{lem}{Lemma}
\newtheorem{rem}{Remark}

\newtheorem{prop}{\it Property}
\newtheorem{prob}{Problem}
\newtheorem{ass}{Assumption}
\newtheorem{obj}{Objective}

\usepackage{tikz}
\usetikzlibrary{positioning}
\usetikzlibrary{shapes,arrows}
\newcommand{\col}{\mbox{col}}

\allowdisplaybreaks
\begin{document}
\title{A nonparametric learning framework for nonlinear robust output regulation}

\author{Shimin~Wang, Martin~Guay, Zhiyong Chen, Richard D. Braatz

\thanks{{\myr This research was supported by the U.S. Food and Drug Administration under the FDA BAA-22-00123 program, Award Number 75F40122C00200.}
The work is also supported by NSERC. 

Shimin Wang and Richard D. Braatz are with Massachusetts
Institute of Technology, USA.

Martin Guay is with Queen's University, Kingston, Ontario, K7L 3N6, Canada. 

Zhiyong Chen is with The University of Newcastle, Callaghan, NSW 2308, Australia.

E-mail: bellewsm@mit.edu, braatz@mit.edu, guaym@queensu.ca, zhiyong.chen@newcastle.edu.au.
 }}

\maketitle

\begin{abstract}\myr
A nonparametric learning solution framework is proposed for the global nonlinear robust output regulation problem. 
%
%
%
%
We first extend the assumption that the steady-state generator is linear in the exogenous signal to the more relaxed assumption that it is polynomial in the exogenous signal.
Additionally, a nonparametric learning framework is proposed to eliminate the construction of an explicit regressor, as required in the adaptive method, which can potentially simplify the implementation and reduce the computational complexity of existing methods.
With the help of the proposed framework, the robust nonlinear output regulation problem can be converted into a robust non-adaptive stabilization problem for the augmented system with integral input-to-state stable (iISS) inverse dynamics. 
Moreover, a dynamic gain approach can adaptively raise the gain to a sufficiently large constant to achieve stabilization without requiring any a priori knowledge of the uncertainties appearing in the dynamics of the exosystem and the system.
Furthermore, we apply the nonparametric learning framework to globally reconstruct and estimate multiple sinusoidal signals with unknown frequencies without the need for adaptive parametric techniques.
An explicit nonlinear mapping can directly provide the estimated parameters, which will exponentially converge to the unknown frequencies.
Finally, a feedforward control design is proposed to solve the linear output regulation problem using the nonparametric learning framework. 
%
%
Two simulation examples are provided to illustrate the effectiveness of the theoretical results.
\end{abstract}


\section{Introduction}\label{Sec-Intro}%
The internal model principle is essential for the design of control systems that achieve asymptotically tracking and disturbance rejection \cite{francis1977linear} to solve the classically named output regulation problem \cite{isidori1990output,huang1990nonlinear}.
In particular, the linear internal model principle successfully turns the linear output regulation problem into a pole placement problem \cite{huang2004general}. 
In the nonlinear case, the nonlinear superposition of the parameter uncertainties and nonlinear dynamics behaviour arising from the system, reference inputs, and disturbances generate a nonlinear steady-state behaviour of the closed-loop system. The presence of nonlinearities prevents the direct application of the linear internal model and feedforward controller design approach \cite{huang1994robust}.
Specifically, the results in \cite{huang2001remarks} identified a link between a polynomial and a partial differential equation conditions on an input feedforward function, resulting in a sufficient condition called trigonometric polynomial for the solvability of the nonlinear robust output regulation problem. 
{\myr The construction of a proper internal model approach that solves the output regulation problem for various system dynamics over diverse conditions remains an active research area in both the automatic control literature \cite{byrnes2003limit,wu2021output,bin2022internal,huang2021internal,wang2022robust,lu2019adaptive} and neuroscience \cite{broucke2022adaptive}}.

The existing research addressing output regulation problem can be classified following the internal models type \cite{bin2022internal} which include nonlinear internal models \cite{byrnes2003limit,huang2004general}, the canonical linear internal model \cite{nikiforov1998adaptive,serrani2001semi,liu2009parameter,broucke2024disturbance,lu2019adaptive} and a generic internal model \cite{marconi2008uniform,xu2019generic}.
To illustrate, two different nonlinear internal models were given in \cite{byrnes2003limit} and \cite{huang2004general} to handle the nonlinear output regulation problem with linear output mapping immersed in a known nonlinear steady-state generator.
As a matter of fact, an important class of linear and nonlinear output regulation problems is for systems with an uncertain exosystem with unknown parameters or dynamics \cite{serrani2001semi,liu2009parameter,marino2011adaptive,isidori2012robust}.
To handle the unknown parameters arising from the uncertain exosystem, the canonical linear internal model \cite{nikiforov1998adaptive,serrani2001semi,chen2002global,liu2009parameter} and the generic internal model \cite{marconi2008uniform,xu2019generic} are the two main approaches considered in the literature. 
The canonical linear internal model was first proposed in \cite{nikiforov1998adaptive} to overcome important shortcomings of the classical internal model, such as the need for exact models of the controlled system plant and the disturbance generator, by using a parameterized method and adaptive control techniques.  
The technique has been applied for the solution of linear output regulation, \cite{marino2003output}, and nonlinear output regulation problems, \cite{serrani2001semi,liu2009parameter}, for systems with a linear uncertain exosystem. 
Interestingly, it is shown in \cite{liu2009parameter} that the estimated parameter vector will converge to the actual parameter vector if the canonical internal model has a dimension that is no more than twice the number of sinusoidal signals in the steady-state input of the system.
%
%
%
Additionally, the construction of parameterized nonlinear internal models using a canonical linear internal model approach has been proposed in \cite{xu2017constructive} to help address nonlinear output regulation problems based on Kreisselmeier’s adaptive observers \cite{kreisselmeier1977adaptive}.
%
%

A generic internal model design was initially introduced in \cite{marconi2008uniform} to relax the various assumptions imposed on internal model candidates as mentioned above \cite{nikiforov1998adaptive,byrnes2003limit,huang2004general}. 
The generic internal model has been employed to solve a practical semi-global nonlinear output regulation problem in \cite{marconi2008uniform}. This important contribution sparked 
research interest in the use of non-adaptive techniques for the solution of nonlinear output regulation problems \cite{marconi2007output,isidori2012robust,xu2019generic}.
The non-adaptive control approach naturally avoids the burst phenomenon \cite{anderson1985adaptive,hsu1987bursting} that is encountered in adaptive iterative learning methods, such as in the MIT rule \cite{mareels1987revisiting}. 
{\myr In addition, the non-adaptive method in \cite{marconi2007output,isidori2012robust,xu2019generic}  avoids the need for the construction of a Lyapunov function or input-to-state stability (ISS) Lyapunov function for the closed-loop system with positive semidefinite derivatives.}
%
In fact, a counterexample was used in \cite{chen2023lasalle} to show that the boundedness property cannot be guaranteed when the derivative of an ISS Lyapunov function is negative semidefinite even with a small input term. 
%
%
%
%
Remarkably, the generic internal model in \cite{marconi2007output,marconi2008uniform,xu2019generic} could directly estimate the unknown parameters in the exosystem without using the parameterized method and adaptive control technique, significantly improving the canonical linear internal model \cite{bin2022internal}.
Accordingly, designing a generic internal model for solving output regulation problems has received considerable attention in the literature \cite{marconi2007output,marconi2008uniform,xu2019generic}. 
The generic internal model relies on the explicit construction of a nonlinear continuous mapping function, which is known to exist \cite{KE2003}. 
Methods for constructing a locally Lipschitz approximation of such a function were discussed in \cite{marconi2008uniform}.
In particular, a nonlinear high-gain observer approach was constructed in \cite{isidori2012robust} to remove the ‘‘explicit’’ need for adaptive techniques, as proposed in \cite{marino2011adaptive}, by assuming that the steady-state generator is linear in the exogenous signal and that the real Hankel matrix is invertible.
Moreover, approximations of the nonlinear output mapping function were introduced in \cite{bin2020approximate,bernard2020adaptive} to provide an adaptive online tuning of the regulator by employing system identification algorithms selecting the optimized parameters according to a specific least-squares policy.
%
Subsequently, nonlinear output mapping functions were explicitly constructed in \cite{xu2019generic}  under the assumption that the steady-state generator is linear in the exogenous signal to exploit the known property that the steady-state generator of nonlinear output regulation is of finite $k^\text{th}$ order in the exogenous signal \cite{huang1992approximation,huang1994robust}.

The internal model can be viewed as an observer of the steady-state generator providing an online estimated steady-state/input. The observability condition of the steady-state generator has been discussed in \cite{chen2010observability}.
The main differences between feedback and feedforward control frameworks are related to the ability to measure the output of the steady-state generator. When a direct measurement is available an internal model-based design approach would prevail. If only indirect measurements are possible, observer-based designs are required.
As a result, finding an explicit nonlinear output mapping for general internal models also promotes studying the parameter/frequency and state estimation problems.
Some efforts have been made to construct an explicit nonlinear output mapping to estimate the frequency of a pure sinusoidal signal non-adaptively locally using a nonlinear Luenberger observer approach  \cite{praly2006new}.
As highlighted in \cite{afri2016state}, the nonlinear Luenberger observer design provides significant advantages when compared to traditional adaptive observer designs since they yield lower dimensional systems with dynamics that admit a strict Lyapunov function candidates. 
Online parameter/frequency estimation of multiple sinusoidal signals has long been a defining topic in automatic control, communication, image processing and power engineering.
A large majority of observer-based designs for online parameter/frequency estimation utilize traditional adaptive control methods \cite{marino2002global,xia2002global,hou2011parameter,pin2019identification}. 

Motivated by the aforementioned pioneering research, we propose a nonparametric learning solution framework for a generic internal model design of nonlinear robust output regulation. 
We show that the global nonlinear robust output regulation problem for a class of nonlinear systems with output feedback subject to a nonlinear exosystem can be tackled by constructing a generic linear internal model as in \cite{marconi2008uniform,isidori2012robust}, provided that a sufficiently continuous nonlinear mapping exists. 
An explicit continuous nonlinear mapping was constructed recently in \cite{xu2019generic} under the assumption that the steady-state generator is linear in the exogenous signal.
We further relax this assumption to address a polynomial steady-state generator in the exogenous signal, demonstrating that the persistency of excitation condition remains necessary for the generic internal model design following results in \cite{liu2009parameter}.
{\myr A nonparametric learning framework is proposed to solve a linear time-varying equation to ensure that a nonlinear continuous mapping always exists.}
With the help of the proposed nonparametric learning framework, the nonlinear robust output regulation problem can be converted into a robust non-adaptive stabilization problem for the augmented system with integral Input-to-State Stable (iISS) inverse dynamics. 
{\myr One key distinction from the adaptive methods in \cite{nikiforov1998adaptive,serrani2001semi,chen2002global,liu2009parameter} is that the proposed nonparametric framework does not rely on a known explicit regressor, typically generated from the structure of the composite system, which includes the controlled system and internal model. 
Instead, the proposed nonparametric framework directly learns the unknown parameter vector of the steady-state generator from the generic internal model signal. 
The nonparametric framework eliminates the need for explicit regressor construction, potentially simplifying the implementation and reducing the computational complexity. 
Additionally, the framework allows for more flexible adaptation to various system dynamics without predefined model assumptions \cite{wang2024nonparametric}. }
Moreover, a dynamic gain approach can adjust the gain adaptively to achieve stabilization without requiring knowledge of the uncertainties coupled with the exosystem and the control system.
Different from traditional adaptive control methods as in \cite{marino2002global,xia2002global,hou2011parameter,hsu1999globally,pin2019identification,liu2021exponential}, we apply the nonparametric learning framework to globally estimate the unknown frequencies, amplitudes and phases of the $n$ sinusoidal components with bias independent of any adaptive techniques. 
In contrast, the estimation approach proposed in \cite{praly2006new} only considers the frequency estimation of a single sinusoidal signal, which requires the actual frequency to lie within a known (possibly large) compact set.
{\myr We also show that the explicit nonlinear continuous mapping can directly provide the estimated parameters after a specific time, which yields the exponential convergence to the unknown parameters/frequencies.}
Finally, a feedforward control design is proposed to solve the output regulation using the proposed nonparametric learning framework.

{The rest of this paper is organized as follows. Section \ref{Sec-BM} formulates the problem and objectives. In Section \ref{Sec-Main}, we establish a learning-based framework for solving nonlinear robust output regulation over a polynomial steady-state generator. 
The nonparametric learning framework to globally estimate the unknown frequencies, amplitudes and phases of the $n$ sinusoidal
components with bias independent of using adaptive techniques {\myr is} introduced in Section \ref{feedforwadsection}. {\myr This is followed by the application of the proposed nonparametric learning framework for the design of a feedforward controller that solves the linear output regulation problem.
Two numerical examples are presented in  Section \ref{NumPraEx}. A controlled Lorenz system application is presented first. This is followed by an application to a quarter-car active automotive suspension system subject to unknown road profiles are considered. Section \ref{Sec-Conclusion} concludes the paper.}}

\textbf{Notation:} 
{\myr A function $f:[t_0, \infty)\mapsto \mathds{R}^{m\times n}$ is bounded if there exists a positive constant $c$ such that $\|f(t)\|\leq c$ for all $t\geq t_0$.} $\|\cdot\|$ is the Euclidean norm. $\emph{Id} : \mathds{R}\rightarrow \mathds{R}$ is an identity function. {\myr For $X_i\in \mathds{R}^{n_i\times m}$, let $\col(X_1,\dots,X_N)=[X_{1}^{\top},\dots ,X^{\top}_{N}]^\top$.} For any matrix $ P \in \mathds{R}^{n \times m} $ and $ P = [p_1, \dots p_m]$ where $p_i \in \mathds{R}^m $, $\hbox{vec}(P) = \mbox{col }(p_1, \dots, p_m) $. For any column vector $X \in \mathds{R}^{nq}$ for positive integers $n$ and $q$, $M_n^q (X) = [X_1,X_2,\dots, X_q]$. The Kronecker product is denoted by $\otimes$.
A function $\alpha: \mathds{R}_{\geq 0}\rightarrow \mathds{R}_{\geq 0}$ is of class $\mathcal{K}$ if it is continuous, positive definite, and strictly increasing. $\mathcal{K}_o$ and $\mathcal{K}_{\infty}$ are the subclasses of bounded and unbounded $\mathcal{K}$ functions, respectively. For functions $f_1(\cdot)$ and $f_2(\cdot)$ with compatible dimensions, their composition $f_{1}\left(f_2(\cdot)\right)$ is denoted by $f_1\circ f_2(\cdot)$. For two continuous and positive definite functions $\kappa_1(\varsigma)$ and $\kappa_2(\varsigma)$, $\kappa_1\in \mathcal{O}(\kappa_2)$ means $\limsup_{\varsigma\rightarrow 0^{+}}\frac{\kappa_1(\varsigma)}{\kappa_2(\varsigma)}<\infty$.

\section{Problem Formulation and Objectives}\label{Sec-BM}%
\subsection{System models and problem formulation}

We consider a nonlinear system in strict-feedback normal form with a unity relative degree
\begin{align}\label{Main-sys1}
\dot{z} &= f(z,y,\mu), \notag\\
\dot{y} &= g(z,y,\mu) + b(z,y,\mu)u, \\
e &= y-h(v,w),\notag
\end{align}
where $\col(z,y)\in\mathds{R}^{n_z+1}$ are the state variables of the system, $y\in\mathds{R}$ is the output to be regulated, $e\in\mathds{R}$ is tracking error, $u\in\mathds{R}$ is the control input variable, and $\mu=\col(v,w,\sigma)$ contains the uncertainties $w\in\mathds{R}^{n_{w}}$, the state of the exosystem $v\in\mathds{R}^{n_{v}}$,
and the unknown parameters $\sigma \in\mathds{R}^{n_{\sigma}}$ arising from the exosystem. The smooth exosystem dynamics are given by
\begin{align}\label{Exosys1}
\dot{v} &= s(v,\sigma),
\end{align}
which generates the disturbance and reference signals. The uncertainties $w$ are assumed to belong to a known compact invariant set $\mathds{W}$. Similarly, the exosystem state variables $v$ also are assumed to stay in some known compact invariant set $\mathds{V}$ \footnote{The set $\mathds{V}$ is called compact and invariant for $\dot{v}=s(v,\sigma)$ if it is compact and for each $v(0)\in\mathds{V}$ and $\sigma\in \mathds{S}$, the solution $v(t)\in\mathds{V}$ for all $t\geq0$.} for the exosystem \eqref{Exosys1} where $\sigma$ is in some known compact set $\mathds{S}$. Suppose that all the functions in \eqref{Main-sys1} are globally defined and sufficiently smooth satisfying $f(0,0,\mu)=0$, $g(0,0,\mu)$ and $h(0,w)=0$ for all $\mu\in \mathds{V}\times \mathds{W}\times \mathds{S}$. The function $b(z,y,\mu)$ is lower bounded by a positive constant $b^*$ for any $\mu\in \mathds{V}\times \mathds{W} \times \mathds{S}$ and $\col(z,y)\in\mathds{R}^{n_z+1}$. 


The nonlinear robust output regulation problem of this paper is formulated as follows:
\begin{prob}[Nonlinear Robust
Output Regulation] \label{ldlesp} Given systems (\ref{Main-sys1}) and (\ref{Exosys1}), any compact subsets $\mathds{S}\in \mathds{R}^{n_{\sigma}}$, $\mathds{W}\in \mathds{R}^{n_w}$ and $\mathds{V}\in \mathds{R}^{n_v}$ with $\mathds{W}$ and $\mathds{V}$ containing the origin, design a control law 
such that for all initial conditions $v(0)\in \mathds{V}$, $\sigma \in \mathds{S}$ and $w\in \mathds{W}$, and any initial states $\textnormal{\col}(z(0), y(0))\in\mathds{R}^{n_z+1}$, the solution of the closed-loop system exists and is bounded for all $t\geq 0$,
 and $$\lim\limits_{t\rightarrow\infty}e(t)=0.$$
\end{prob}
\subsection{Standard assumptions} 
The following assumptions are standard to handle Problem \ref{ldlesp} for the systems \eqref{Main-sys1} and \eqref{Exosys1}. 
\begin{ass}[Solvability of the regulator equations] \label{H1}There is a smooth function $\bm{z}^{\star}{\myr \equiv}\bm{z}^{\star}(v,\sigma,w)$ for $\textnormal{\col}(v,w,\sigma)\in\mathds{V}\times \mathds{W} \times \mathds{S}$ satisfying 
\begin{align}\label{Reg-1}
\frac{\partial \bm{z}^{\star}(v,\sigma,w)}{\partial v}s(v,\sigma) &= f(\bm{z}^{\star}(v,\sigma,w),0,\mu),
\end{align}
It determines the so-called output zeroing manifold $\bm{u}{\myr \equiv}\bm{u}(v,\sigma,w)$ given by
\begin{align}\label{Reg-2}
\bm{u}(v,\sigma,w) &=  b(\bm{z}^{\star}(v,\sigma,w), h(v, w), v, \sigma, w)^{-1}\Big(\frac{\partial h(v, w)}{\partial v} s(v,\sigma)\notag\\
&-g(\bm{z}^{\star}(v,\sigma,w),0,v,\sigma,w) \Big),
\end{align}
which can be regarded as an output of the exosystem \eqref{Exosys1}.
\end{ass}

\begin{ass}[Minimum-phase condition]\label{H2}  The translated inverse system 
\begin{align}\label{Reg-2}
\dot{\bar{z}} &= f(\bar{z}+\bm{z}^{\star},y,\mu)-f(\bm{z}^{\star},0,\mu) 
\end{align}
is input-to-state stable with state $\bar{z}=z-\bm{z}^{\star}$ and input $y$ in the sense of \cite{Sontag2019}. In particular,  
there exists a continuous function $V_{\bar{z}}(\bar{z})$ satisfying 
$$\underline{\alpha}_{\bar{z}}\!\left(\|\bar{z} \|\right)\leq V_{\bar{z} }\!\left(\bar{z} \right) \leq \overline{\alpha}_{\bar{z}}\!\left(\|\bar{z} \|\right)$$
for some class $\mathcal{K}_{\infty}$ functions $\underline{\alpha}_{\bar{z}}\left(\cdot\right)$ and $\overline{\alpha}_{\bar{z}}\left(\cdot\right)$ such that, for any $v\in \mathds{V}$, along the trajectories of the $\bar{z} $ subsystem,
$$\dot{V}_{\bar{z} }\leq-\alpha_{\bar{z}}\!\left(\|\bar{z} \|\right)+ \gamma\!\left(y \right),$$
where $\alpha_{\bar{z}}\!\left(\cdot\right)$ is some known class $\mathcal{K}_{\infty}$ function satisfying $\limsup\limits_{\varsigma\rightarrow 0^{+}}\left(\alpha_{\bar{z}}^{-1}\!\left(\varsigma^2\right)/\varsigma\right)< + \infty$, and $\gamma\!\left(\cdot\right)$ is some known smooth positive definite function.
\end{ass}
\subsection{Linear generic internal model design and error system dynamics} 
Based on Assumption \ref{H1}, following the pioneering work in \cite{marconi2007output,marconi2008uniform,xu2019generic,bin2020approximate},  we can construct a linear generic internal model in the form 
\begin{align}\label{IM-01}
\dot{\eta} &= M\eta + Nu,~~ \eta\in\mathds{R}^{n_{0}},
\end{align}
where $(M,N)$ is controllable and $M$ is Hurwitz. As shown in \cite{marconi2007output,marconi2008uniform,xu2019generic}, this can be done by using the result of \cite{KE2003}. In fact, let 
\begin{align*}
\bm{\eta}^{\star}(v(t),\sigma,w)=&\int_{-\infty}^{t} e^{M(t-\tau)}N\bm{u}(v(\tau), \sigma, w) d\tau,\\
\bm{u}(v(t), \sigma, w)=&\chi(\bm{\eta}^{\star}(v(t),\sigma,w)), ~~ \bm{\eta}^{\star}\in\mathds{R}^{n_{0}}
\end{align*}
{\myr along each system trajectory associated with the initial condition, parameter and uncertainty vector $\col(v(0),\sigma,w)$, that will satisfy} the  differential equations 
\begin{align}\label{IM-01}
\frac{d \bm{\eta}^{\star}(v(t),\sigma,w)}{dt}  &= M\bm{\eta}^{\star}(v(t),\sigma,w) + N\bm{u}(v(t), \sigma, w), \notag\\
\bm{u}(v(t), \sigma, w) &= \chi(\bm{\eta}^{\star}(v(t),\sigma,w)).
\end{align}
 These trajectories are those of the steady-state generator of control input $\bm{u}(v(t), \sigma, w)$ for a sufficiently large dimension $n_0$, {\myr a continuous mapping $\chi(\cdot)$ and the initial condition $\bm{\eta}^{\star}(v(0),\sigma,w)$.}  For convenience, let $\bm{\eta}^{\star}{\myr \equiv}\bm{\eta}^{\star}(v,\sigma,w)$.
{\myr To solve Problem \ref{ldlesp},} we perform the coordinate/input transformations
\begin{align*}
\bar{z}&=\, z-\bm{z}^{\star}, \\
\bar{\eta} &=\,\eta-\bm{\eta}^{\star} -b(\bar{z}+\bm{z}^{\star},e+h(v, w),\mu)^{-1}Ne. 
\end{align*}
Using the transformation, the augmented system 
\begin{subequations}\label{augmen-1}
\begin{align}
\dot{\bar{z}} =&\,\bar{f}(\bar{z},e,\mu) \label{augmen-1a}\\
 \dot{\bar{\eta}} =&\,M\bar{\eta}+\bar{p}(\bar{z},e,\mu) \label{augmen-1b}\\
 \dot{e} =& \, \bar{g}(\bar{z},\bar{\eta},e,\mu)+ \bar{b}(\bar{z},e,\mu)(u-\chi( \bm{\eta}^{\star})) \label{augmen-1c}
\end{align}
\end{subequations}
is obtained, where $\bar{b}(\bar{z},e,\mu)=b(\bar{z}+\bm{z}^{\star},e+h(v, w),\mu)$,
\begin{align*}
    \bar{f}(\bar{z},e,\mu)=&\,f(\bar{z}+\bm{z}^{\star},e+h(v, w),{\myr \mu})-f(\bm{z}^{\star},{\myr h(v, w)},\mu),\\
    \bar{p}(\bar{z},e,\mu)=&\,\bar{b}(\bar{z},e,\mu)^{-1}\left(MNe-N\bar{g}(\bar{z},e,\mu)\right)\\
    &-\frac{d \bar{b}(\bar{z},e,\mu)^{-1}}{dt}Ne,\\
     \bar{g}(\bar{z},\bar{\eta},e,\mu)=&\,g(\bar{z}+\bm{z}^{\star},e+h(v, w),\mu) -g(\bm{z}^{\star},{\myr h(v, w)},\mu).
\end{align*}
For the system \eqref{augmen-1}, under Assumptions \ref{H1} and \ref{H2}, we have the following properties: 
\begin{prop}\label{property1}%
There exists a smooth input-to-state Lyapunov function $V_0{\myr \equiv}V_0(\bar{Z})$ satisfying
\begin{align}
\underline{\alpha}_0(\|\bar{Z}\|^2)&\leq V_0(\bar{Z})\leq \bar{\alpha}_0(\|\bar{Z}\|^2),\notag\\
\dot{V}_0\big|_{\eqref{augmen-1a}+\eqref{augmen-1b}} &\leq  -\|\bar{Z}\|^2+\bar{\gamma} \left(e\right),\label{V0}
\end{align}
for some comparison functions $\underline{\alpha}_0(\cdot)\in {\myr \mathcal{K}_{\infty}}$, $\bar{\alpha}_0(\cdot)\in {\myr \mathcal{K}_{\infty}}$ and {\myr a smooth positive definite function,  $\bar{\gamma}(\cdot)$,} with $\bar{Z}=\textnormal{\col}(\bar{\eta},\bar{z})$.
\end{prop}

\begin{prop}\label{property2}%
There are positive smooth functions $\gamma_{g0}(\cdot)$ and $\gamma_{g1}(\cdot)$ such that
$$|\bar{g}(\bar{Z}, e, \mu)|^2\leq \gamma_{g0}(\bar{Z})\|\bar{Z}\|^2+e^2\gamma_{g1}(e),$$
{\myr for $\mu\in \mathds{V}\times \mathds{W}\times \mathds{S}$.}
\end{prop}

The proof of Properties~\ref{property1} and \ref{property2} can be found in Appendix~\ref{AppendixA}.

\begin{thm}\label{Theorem-1}%
For the composite system \eqref{Main-sys1} and \eqref{Exosys1} under Assumptions \ref{H1} and \ref{H2}, there is a positive smooth function $\rho(\cdot)$, positive number $k^*$ and the controller,  
\begin{subequations}\label{ESC-1}%
\begin{align}
\dot{\eta} &=M\eta+Nu,\label{ESC-1a}\\
u &= -k {\rho}(e)e + \chi(\eta),\label{ESC-1b}
\end{align}\end{subequations}
{\myr solves Problem \ref{ldlesp}}. Furthermore, the closed-loop system composed of \eqref{augmen-1} and \eqref{ESC-1} has the
property that there exists a continuous
positive definite function $V{\myr \equiv}V(\bar{Z}, e )$
such that, for all $\mu\in \mathds{S}\times \mathds{V}\times \mathds{W}$,
$$\dot{V}\leq -\big\|\bar{Z}\big\|^2-e^2$$
for any $k\geq k^*$. 
\end{thm}%
The proof of Theorem~\ref{Theorem-1} can be found in   Appendix~\ref{AppendixC}.

\begin{cor}\label{Theorem-2}%
For the composite system \eqref{Main-sys1} and \eqref{Exosys1} under Assumptions \ref{H1} and \ref{H2}, there is a positive smooth function $\rho(\cdot)$ such that the controller
\begin{subequations}\label{ESC-2}%
\begin{align}
\dot{\eta} &=M\eta+Nu,\label{ESC-2a}\\
\dot{\hat{k}}&={\rho}(e)e^2,\\
u &= -\hat{k} {\rho}(e)e + \chi(\eta),\label{ESC-2b}
\end{align}\end{subequations}
solves Problem \ref{ldlesp}. 
\end{cor}%
\begin{rem}The proof of Corollary \ref{Theorem-2} easily follows from the Lyapunov function 
$$U_0(\bar{Z},e, \hat{k}-k^*)=V(\bar{Z}, e)+(1/2)(\hat{k}-k^*)^2$$
and is omitted for the sake of brevity.
\end{rem}
\subsection{Objectives} 
{  Based on the discussion above, it follows that the {\it Nonlinear Robust Output Regulation Problem} \ref{ldlesp} can be solved by \eqref{ESC-1} if we can find an explicit continuous mapping $\chi(\cdot)$. However, the result in \cite{KE2003} is limited to a proof of the existence of the nonlinear mapping $\chi(\cdot)$. Methods for the construction of a locally Lipschitz approximation of such a function were discussed in \cite{marconi2008uniform}.
The nonlinear output mapping function was approximated in \cite{bin2020approximate,bernard2020adaptive}  by introducing an adaptive online tuning of the regulator that employs system identification algorithms to the optimized parameters according to a specific least-squares policy.}
{The explicit construction of a nonlinear continuous output mapping function was proposed in \cite{xu2019generic} under the assumption that the steady-state generator is linear in the exogenous signal. The steady-state generator of the nonlinear output regulation system of order $k$ in the exogenous signal was proposed and analyzed in  \cite{huang1992approximation,huang1994robust}.
The first objective of this study is as follows.

\begin{obj} \label{obj1} For system \eqref{Main-sys1} and \eqref{Exosys1} with the nonlinear mapping $\chi(\cdot)$ in \eqref{IM-01}, we seek the explicit construction of a nonlinear continuous output mapping function $\chi(\cdot)$ under the assumption that the steady-state generator is polynomial in the exogenous signal. 
\end{obj}


{\myr One of the distinctive features of the generic internal model \eqref{IM-01} is that the unknown parameter vector $a(\sigma)$ determined by $\sigma$ can be identified by the signal $\bm{\eta}^{\star}$ such that  $a(\sigma)=\check{a}(\bm{\eta}^{\star})$ with $\check{a}(\cdot)$ being a continuous mapping.  As a result, the continuous mapping $\chi(\bm{\eta}^{\star})$ can be  rewritten as $\chi (\bm{\eta}^{\star})=\chi (\check{a}(\bm{\eta}^{\star}), \bm{\eta}^{\star})$}. The nonlinear mapping $\chi \left(\bm{\eta}^{\star}, \check{a}(\bm{\eta}^{\star})\right)$ and $\check{a}(\bm{\eta}^{\star})$ should be well-defined, continuous and such that the control law \eqref{ESC-1} can generate the required control signal online and in real-time. It is noted that $\chi(\cdot)$ and $\check{a}(\cdot)$ are not always well-defined and continuous for $\eta(t)$ over $t\geq 0$ due to some conditions not being satisfied by $\eta(t)$ over $t\geq 0$. Hence, the second objective of this paper is below.

\begin{obj} \label{obj2} For system \eqref{Main-sys1} and \eqref{Exosys1} with controller
\eqref{ESC-1} and nonlinear mapping $\chi(\cdot)$ given in \eqref{IM-01}, we seek a nonparametric learning framework to learn the unknown parameter $a(\sigma)$ that provides a well defined, continuous mapping $\chi(\cdot)$ for $\eta(t)$ over $t\geq 0$ solving Problem \ref{ldlesp}. 
\end{obj}
\section{A Nonparametric Learning Framework for Nonlinear Robust Output Regulation}\label{Sec-Main}%
To reduce the complexity of the problem and achieve our two objectives, for convenience, we simplify the nonlinear exosystem case into a linear exosystem and introduce the standard assumptions commonly seen in nonlinear robust output regulation as in \cite{liu2009parameter}. 
\subsection{An explicit nonlinear mapping for the polynomial steady-state generator}
\begin{ass}\label{ass1-explicit}The dynamics of exosystem system \eqref{Exosys1} is $ s(v,\sigma) =S(\sigma)v$. All the eigenvalues of $S(\sigma)$ are distinct with zero real part for all $\sigma\in \mathds{S}$.
\end{ass}
\begin{ass}\label{ass4-explicit} The functions $\bm{u}(v,\sigma,w)$ are polynomials in $v$ with coefficients depending on $\sigma$ for all $\sigma\in \mathds{S}$.
\end{ass}
\begin{rem}\label{remPE}
From \cite{huang2001remarks,liu2009parameter}, under Assumption \ref{ass4-explicit}, for the function $\bm{u}(v,\sigma,w)$, there is an integer {\color{red}$n^*>0$} such that $\bm{u}(v,\sigma,w)$ can be expressed by
$$\bm{u}(v(t),\sigma,w)=\sum_{j=1}^{n^*}C_{j}(v(0), w,\sigma)e^{\imath {\omega}_jt},$$
for some functions $C_{j}(v(0), w,\sigma)$, where $\imath$ is the imaginary unit and ${\omega}_j$ are distinct real numbers for $0\leq j \leq n^*$.
\end{rem}
\begin{ass}\label{ass5-explicit}  
 For any $v(0)\in \mathds{V}$, $w\in \mathds{W}$ and $\sigma\in \mathds{S}$, $C_{j}(v(0), w,\sigma)\neq 0$.
\end{ass}
As shown in \cite{huang2004nonlinear}, under Assumptions \ref{ass1-explicit} and \ref{ass4-explicit}, there exists a positive integer $n$, such that $\bm{u}(v,\sigma,w)$ satisfies for all $\col(v,w,\sigma)\in \mathds{V}\times\mathds{W}\times \mathds{S}$,
\begin{align} \label{aode-explicit}\frac{d^{n}\bm{u}(v,\sigma,w)}{dt^{n}}+a_{1}&(\sigma)\bm{u}(v,\sigma,w)+a_{2}(\sigma)\frac{d\bm{u}(v,\sigma,w)}{dt}\notag\\
+&\dots+a_{n}(\sigma)\frac{d^{n-1}\bm{u}(v,\sigma,w)}{dt^{n-1}}=0,
\end{align}
where $a_{1}(\sigma),\dots,a_{n-1}(\sigma)$ and $ a_{n}(\sigma)$ {\myr all belong to} $\mathds{R}$. Under Assumptions \ref{ass1-explicit} and \ref{ass4-explicit}, equation \eqref{aode-explicit} is such that the polynomial $$P(\varsigma)=\varsigma^{n}+a_{1}(\sigma)+a_{2}(\sigma)\varsigma+\dots+a_{n}(\sigma)\varsigma^{n-1}$$ has distinct roots with zeros real parts for all $\sigma\in \mathds{S}$. Let $a(\sigma)=\col(a_{1}(\sigma), \dots, a_{n}(\sigma))$, $\bm{\xi}(v,\sigma,w)= \col\!\left(\bm{u}(v,\sigma,w),\frac{d\bm{u}(v,\sigma,w)}{dt},\dots,\frac{d^{n-1}\bm{u}(v,\sigma,w)}{dt^{n-1}}\right)$ and $\bm{\xi}{\myr \equiv}\bm{\xi}(v,\sigma,w)$, and define
\begin{align*}
  \Phi(a(\sigma)) =&\left[
                      \begin{array}{c|c}
                        0_{(n-1)\times 1} & I_{n-1} \\
                        \hline
                        -a_{1}(\sigma) &-a_{2}(\sigma),\dots,-a_{n}(\sigma) \\
                      \end{array}
                    \right], \\
  \Gamma =&\left[
                \begin{array}{cccc}
                  1 & 0 & \cdots{} &0 \\
                \end{array}
              \right]_{1\times n}.
\end{align*}
Then, for the defined matrices $\Phi(a(\sigma))$ and $\Gamma$, $\bm{\xi}(v,\sigma,w)$ satisfies
\begin{subequations}\label{stagerator}\begin{align}
\dot{\bm{\xi}}(v,\sigma,w)=& \Phi (a(\sigma))  \bm{\xi}(v,\sigma,w),\\
   \bm{u}(v,\sigma,w)=&  \Gamma \bm{\xi}(v,\sigma,w).
\end{align}\end{subequations}
System \eqref{stagerator} is called a steady-state generator with output $u$, which can be used to generate the steady-state input signal $\textbf{u}(v,\sigma,w)$. 
A major source of inspiration in deriving
our result comes from the work in \cite{xu2019generic}. 

Define the matrix pair $(M, N)$ in \eqref{IM-01} by
\begin{subequations}\label{MNINter}\begin{align}
M=&\left[
                      \begin{array}{c|c}
                        0_{(2n-1)\times 1} & I_{2n-1} \\
                        \hline
                        -m_{1} &-m_{2},\dots,-m_{2n} \\
                      \end{array}
                    \right],\\
N=&\left[
                \begin{array}{ccccc}
                  0 & 0 & \cdots{} &0 &1 \\
                \end{array}
              \right]_{1\times 2n}^\top,                   
\end{align}
\end{subequations}
where $m_{1}$, $m_{2}$ , $\dots$, $m_{2n}$ are chosen such that $M$ is a Hurwitz matrix.

To synthesize the nonlinear operator $\chi (\cdot)$ in \eqref{IM-01}, we recall and reformulate some lemmas stated in \cite{xu2019generic}. They are listed as follows:
\begin{lem}\label{xu2019generic-explicit}\cite{xu2019generic}
Under Assumptions \ref{H1}, \ref{ass1-explicit} and \ref{ass4-explicit}, the matrix-valued function
$$\Xi(a)\ {\myr \equiv} \ \Phi(a)^{2n }+\sum\nolimits_{j=1}^{2n }m_{j}\Phi(a)^{j-1} \in \mathds{R}^{n \times n}$$ is non-singular and 
\begin{align}\label{XIQA-explicit}
\Xi(a)^{-1} =\textnormal{\col}(Q_{1}(a),\dots,Q_{n}(a))\in \mathds{R}^{n \times n},
\end{align}
with $Q_{j}(a)=\Gamma \Xi (a)^{-1}\Phi (a )^{j-1} \in \mathds{R}^{1\times n }$, $j=1,\dots, n$.
\end{lem}
Define the real Hankel matrix (see \cite{afri2016state}): 
\begin{align}
\Theta (\theta)\, & {\myr \equiv} \left[\begin{matrix}\theta_{1} &\theta_{2}&\cdots{}&\theta_{n}\\
\theta_{2}&\theta_{3}&\cdots{}&\theta_{n+1}\\
\vdots&\vdots&\ddots&\vdots\\
\theta_{n} &\theta_{n +1}&\cdots{}&\theta_{2n -1}
\end{matrix}\right]\!\in \mathds{R}^{n \times n },\nonumber
\end{align}
where $\theta =Q \bm{\xi} $, $Q_{j}(a)=\Gamma \Xi(a)^{-1}\Phi(a )^{j-1} \in \mathds{R}^{1\times n }$, $1\leq j\leq 2n $ with
\begin{align}\label{Qdefini}
Q \, {\myr \equiv} \, \textnormal{\col}(Q_{1},\dots,Q_{2n})\in \mathds{R}^{2n\times n}.
\end{align}
{\myr We reformulate Theorem 3.1 in \cite{xu2019generic} to show that matrices $Q$, $M$, $N$, and $\Gamma$ satisfy the matrix equation}:
 {\myr \begin{align}
M Q =&\ \textnormal{\col}\Big(Q_{2},\dots,Q_{2n},-\sum\limits_{j=1}^{2n }m_{j}Q_{j}\Big)\notag\\
=&\ \textnormal{\col}\!\left(Q_{2},\dots,Q_{2n},  \Gamma \Xi(a) ^{-1}\left[-\sum\nolimits_{j=1}^{2n }m_{j}\Phi (a )^{j-1}\right]\right).\notag
\end{align}
Then, from Lemma \ref{xu2019generic-explicit}, we have
\begin{align}
M Q=&\ \textnormal{\col}\Big(Q_{2},\dots,Q_{2n},  \Gamma \Xi (a)^{-1}\big[\Phi (a )^{2n }-\Xi(a) \big]\Big)\notag\\
=&\ \textnormal{\col}\Big(Q_{2},\dots,Q_{2n},  \Gamma \Xi(a) ^{-1}\Phi (a )^{2n }-\Xi^{-1}(a)\Xi(a)\Gamma \Big)\notag\\
=&\ \textnormal{\col}\Big(Q_{2},\dots,Q_{2n},  Q_{2n}\Phi (a )-\Gamma \Big)\notag\\
=&\ \textnormal{\col}\Big(Q_{1}\Phi (a),\dots,Q_{2n-1}\Phi (a), Q_{2n}\Phi (a)-\Gamma \Big)\notag\\
=&\ Q \Phi(a)-N\Gamma, \label{MNGAMMAPhi}
\end{align}
which is called the generalized Sylvester equation, with the explicit solution being discussed in \cite{zhou2005explicit}.} 

To remove the assumption that the steady-state generator is linear in the exogenous signal in \cite{isidori2012robust,xu2019generic}, we now introduce a Lemma that establishes a connection between the persistent excitation condition and the generic internal model.


\begin{lem}\label{theta-invertable-explicit}
Under Assumptions \ref{H1}, \ref{ass1-explicit}, \ref{ass4-explicit}, and \ref{ass5-explicit}, the matrix-valued functions $\left[\begin{matrix}\bm{\xi} &\Phi (a )  \bm{\xi}&\cdots{} & \Phi ^{n -1}(a )  \bm{\xi}\end{matrix}\right]$ and $\Theta (\theta )$ are non-singular for all $t\geq 0$.
\end{lem}
\begin{proof}{\myr It is noted from \eqref{Qdefini} that 
\begin{align*}
\theta &=\col(\theta_1,\dots,\theta_{2n})\\
&=\textnormal{\col}(Q_{1}\bm{\xi},Q_{2}\bm{\xi}, \dots,Q_{2n}\bm{\xi}) 
\end{align*}
with $Q_{j} =\Gamma \Xi ^{-1}\Phi ^{j-1}$, $1\leq j\leq 2n $.
Then, the  real Hankel matrix $\Theta (\theta )$ admits the following equations:}
\begin{align}\label{theta-seperate-explicit}
\Theta (\theta )=&\left[\begin{matrix}Q_{1}\bm{\xi}  &Q_{2}\bm{\xi} &\cdots{}&Q_{n}\bm{\xi} \\
Q_{2}\bm{\xi} &Q_{3}\bm{\xi} &\cdots{}&Q_{n +1}\bm{\xi} \\
\vdots&\vdots&\ddots&\vdots\\
Q_{n}\bm{\xi}  &Q_{n +1}\bm{\xi} &\cdots{}&Q_{2n -1}\bm{\xi} 
\end{matrix}\right]\nonumber\\
=&\left[\begin{matrix}Q_{1}\bm{\xi}  &Q_{1}\Phi \bm{\xi} &\cdots{}&Q_{1}\Phi^{n-1} \bm{\xi} \\
Q_{2}\bm{\xi} &Q_{2}\Phi \bm{\xi} &\cdots{}&Q_{2}\Phi ^{n-1}\bm{\xi} \\
\vdots&\vdots&\ddots&\vdots\\
Q_{n}\bm{\xi}  &Q_{n}\Phi \bm{\xi} &\cdots{}&Q_{n}\Phi ^{n-1}\bm{\xi}
\end{matrix}\right]\nonumber\\
=&\ \underbrace{\col(Q_{1},\dots,Q_{n})}_{\Xi(a)^{-1}}\!\left[\begin{matrix}\bm{\xi}  &\Phi \bm{\xi} &\cdots{} & \Phi^{n-1}  \bm{\xi} \end{matrix}\right].
\end{align}
Under Assumptions \ref{H1}, \ref{ass1-explicit}, \ref{ass4-explicit}, and \ref{ass5-explicit}, the vector $\bm{\xi}$ is persistently exciting following Theorem 4.1 in \cite{liu2009parameter}. Then, from Theorem 1 in \cite{padoan2016geometric} or Lemma 2 in \cite{padoan2017geometric}, the 
 steady-state generator \eqref{stagerator} satisfies the excitation rank condition at every $\bm{\xi}$, which further implies that $\textnormal{rank}\left(\left[\begin{matrix}\bm{\xi}  &\Phi \bm{\xi} &\dots & \Phi ^{n -1}  \bm{\xi} \end{matrix}\right]\right)=n $ and that the matrix $\left[\begin{matrix}\bm{\xi}  &\Phi \bm{\xi} &\dots & \Phi ^{n -1}  \bm{\xi} \end{matrix}\right]$ is invertible. Therefore, $\Theta (\theta )$ is nonsingular from \eqref{theta-seperate-explicit} and Lemma \ref{xu2019generic-explicit}.
\end{proof}

\begin{rem}
 Lemma 3.2 in \cite{xu2019generic} achieved similar results, assuming that the steady-state generator is linear in the exogenous signal. In contrast, Lemma \ref{theta-invertable-explicit} only requires the steady-state generator to be polynomial in the exogenous signal, which is an improvement over \cite{xu2019generic}.
 In particular, \cite{isidori2012robust} also defined a real Hankel matrix by assuming that the steady-state generator is linear in the exogenous signal and the real Hankel matrix is invertible. 
  Lemma \ref{theta-invertable-explicit} links the results in Theorem 4.1 in \cite{liu2009parameter} with the generic internal approach, demonstrating that the persistency of excitation remains a necessary condition for the generic internal model approach.  
 {\myr It is well known that the steady-state generator of nonlinear output regulation is k$^{\text th}$-order in the exogenous signal \cite{huang1992approximation,huang1994robust}. Detailed relations which yield a steady-state generator that is linear or polynomial in the exogenous signal can be found in \cite{huang1995asymptotic,huang2001remarks}. }
 \end{rem}

Using Lemma \ref{theta-invertable-explicit}, we can remove the assumption that the steady-state generator is linear in the exogenous signal in Theorem 3.1 in \cite{xu2019generic} and directly show that the explicit continuous nonlinear mapping in \cite{xu2019generic} is also applicable under the assumption that the steady-state generator is polynomial in the exogenous signal. The proof is modified from Thm.\ 3.1 in \cite{xu2019generic}.

\begin{lem}\label{Lemmamappingf}
Under Assumptions \ref{H1}, \ref{ass1-explicit}, \ref{ass4-explicit} and \ref{ass5-explicit}, there exists an explicit continuous nonlinear mapping to achieve Objective \ref{obj1} with
$$\chi \left(\bm{\eta}^{\star}, \check{a} (\bm{\eta}^{\star} )\right)=\Gamma \Xi\!\left(\check{a} (\bm{\eta}^{\star} )\right)\textnormal{\col}(\bm{\eta}^{\star}_{1},\dots,\bm{\eta}^{\star}_{n})$$ that fulfills \eqref{IM-01}, where $\check{a} (\bm{\eta}^{\star})$ is solved from the time-varying equation
\begin{align}\label{a-explicit0}\Theta (\bm{\eta}^{\star})\check{a} (\bm{\eta}^{\star})+\textnormal{\col}(\bm{\eta}^{\star}_{n +1},\dots,\bm{\eta}^{\star}_{2n})=0.
\end{align}
\end{lem}
\begin{proof}{\myr 
Under Assumptions \ref{ass1-explicit} and \ref{ass4-explicit}, the matrix $\Phi(a(\sigma))$ admits the following characteristic polynomial $$\varsigma^{n}+a_{1}(\sigma)+a_{2}(\sigma)\varsigma+\dots+a_{n}(\sigma)\varsigma^{n-1}=0. $$
This, together with the Cayley-Hamilton Theorem, implies
\begin{align}\label{phi-a-charac}
-a_{1}(\sigma)I_{n}-\dots-a_{n}(\sigma)\Phi (a(\sigma))^{n -1}=\Phi (a(\sigma))^{n }.
\end{align}
It is noted from \eqref{Qdefini} that 
\begin{align*}
\theta &=\col(\theta_1,\dots,\theta_{2n})\\
&=\textnormal{\col}(Q_{1}\bm{\xi},\dots,Q_{2n}\bm{\xi}) 
\end{align*}
with $Q_{j} =\Gamma \Xi ^{-1}\Phi ^{j-1} $, $1\leq j\leq 2n $. 
Then, separately pre-multiplying the matrices $\Gamma \Xi ^{-1}$, $\Gamma \Xi ^{-1}\Phi$, \dots, $\Gamma \Xi ^{-1}\Phi^{n-1}$ and  post-multiplying the vector $\bm{\xi}$ on both sides of the matrix equation \eqref{phi-a-charac} results in 
\begin{align*}
\Gamma \Xi ^{-1}\Phi ^{n }\bm{\xi} &=-a_{1}\Gamma \Xi ^{-1}\bm{\xi}-\dots-a_{n}\Gamma \Xi ^{-1} \Phi ^{n -1}\bm{\xi}, \\
\Gamma \Xi ^{-1}\Phi ^{n+1 }\bm{\xi} &=-a_{1}\Gamma \Xi ^{-1}\Phi\bm{\xi}-\dots-a_{n}\Gamma \Xi ^{-1} \Phi ^{n}\bm{\xi}, \\
\vdots\;\;\;\;&\;\;\vdots\;\;\;\;\;\;\;\;\;\;\;\;\;\;\;\;\;\vdots \\
\Gamma \Xi ^{-1}\Phi ^{2n-1 }\bm{\xi} &=-a_{1}\Gamma \Xi ^{-1}\Phi^{n-1}\bm{\xi}-\dots-a_{n}\Gamma \Xi ^{-1} \Phi ^{2n-2}\bm{\xi}. 
\end{align*}
The last system of algebraic equations is equivalent to  }
\begin{align*}
\theta_{n +1}&=-a_{1}\theta_{1}-\dots-a_{n}\theta_{n},\\
\theta_{n +2}&=-a_{1}\theta_{2}-\dots-a_{n}\theta_{n +1},\\
\vdots\;\;\;\;\vdots&\;\;\;\;\;\;\;\;\;\;\;\;\;\;\;\;\;\vdots \\
\theta_{2n}&=-a_{1}\theta_{n}-\dots-a_{n}\theta_{2n -1},
\end{align*}
{\myr which can be put into the compact matrix form}
\begin{align}\label{a-theta}
{\myr-\underbrace{\left[\begin{matrix}\theta_{1} &\theta_{2}&\cdots{}&\theta_{n}\\
\theta_{2}&\theta_{3}&\cdots{}&\theta_{n+1}\\
\vdots&\vdots&\ddots&\vdots\\
\theta_{n} &\theta_{n +1}&\cdots{}&\theta_{2n -1}
\end{matrix}\right]}_{
\Theta (\theta )}\underbrace{\left[\begin{matrix}a_1\\ a_2\\\vdots\\ a_n\end{matrix}\right]}_a =\left[\begin{matrix}\theta_{n +1}\\\theta_{n +2}\\\vdots\\\theta_{2n}\end{matrix}\right].}
\end{align}
Under Assumptions \ref{H1}, \ref{ass1-explicit}, \ref{ass4-explicit}, and \ref{ass5-explicit},
from Lemma \ref{theta-invertable-explicit}, {\myr the constant vector $a$ can be calculated from \eqref{a-theta} with}
\begin{align}\label{a-explicit}
a =-\Theta(\theta )^{-1}\col(\theta_{n +1},\dots,\theta_{2n})\, {\myr \equiv}\, \check{a} (\theta).
\end{align}
Since $\col(\theta_1,\dots,\theta_{2n})= \textnormal{\col}(Q_{1}\bm{\xi},\dots,Q_{2n}\bm{\xi})$, we have
\begin{align*}
\col(\theta_{1},\dots,\theta_{n})=&\ {\myr \col(Q_{1}\bm{\xi},\dots,Q_{n}\bm{\xi})}\\
=&\ \col(Q_{1},\dots,Q_{n})\bm{\xi}.
\end{align*}
{Moreover, $\Xi(a)$ is non-singular and $\col(Q_{1},\dots,Q_{n})= \Xi(a)^{-1}$. As a result, we obtain}
\begin{align}\label{xi-solution-ex}
\bm{\xi} &=\underbrace{\left[\col(Q_{1},\dots,Q_{n})\right]^{-1}}_{\Xi (a)}\!\col(\theta_{1},\dots,\theta_{n}).
\end{align}
Hence, from the function $\check{a} (\theta)$ in \eqref{a-explicit}, we have 
\begin{align}\label{solution-chai-explicit}
\bm{u} (v,\sigma,w)&=   \Gamma   \bm{\xi} \left(v,\sigma,w\right)\nonumber\\
&=\Gamma \Xi \left(\check{a} (\theta )\right)\col(\theta_{1},\dots,\theta_{n})\,{\myr \equiv}\,\chi(\theta, \check{a} (\theta ) ),
\end{align}
which is globally defined and smooth.
Using the generalized Sylvester equation \eqref{MNGAMMAPhi}, the steady-state generator \eqref{IM-01} can be transformed into the steady-state generator \eqref{stagerator} which admits the solution
\begin{align*}
&\bm{\eta}^{\star}(v(t),\sigma,w) 
=\int_{-\infty}^{t} e^{M(t-\tau)}N \Gamma   \bm{\xi} (v(\tau),\sigma,w)d\tau\\
=&\int_{-\infty}^{t} e^{M(t-\tau)}(Q \Phi(a)-MQ)  \bm{\xi} (v(\tau),\sigma,w)d\tau\\
=&\int_{-\infty}^{t} \left[e^{M(t-\tau)}Q\dot{ \bm{\xi} }(v(\tau),\sigma,w)- e^{M(t-\tau)}MQ\bm{\xi} (v(\tau),\sigma,w)\right]\!d\tau\\
=&\, e^{Mt}\int_{-\infty}^{t} \frac{d (e^{-M\tau}Q{ \bm{\xi} }(v(\tau),\sigma,w))}{d\tau}d\tau\\
=&\, e^{Mt}  e^{-M\tau}Q{ \bm{\xi} }(v(\tau),\sigma,w)\Big|_{-\infty}^{t}\\
=&\, e^{Mt}  e^{-Mt}Q{ \bm{\xi} }(v(t),\sigma,w)-\underbrace{\lim_{\tau \rightarrow -\infty}e^{Mt}  e^{-M\tau}Q{ \bm{\xi} }(v(\tau),\sigma,w)}_{0}\\
=&\, Q{ \bm{\xi} }(v(t),\sigma,w).
\end{align*}
This implies that $\theta=\bm{\eta}^{\star}$, completing the proof. 
\end{proof}


\begin{thm}\label{cor-Lemma-explicit}%
For the composite system \eqref{Main-sys1} and \eqref{Exosys1} under Assumptions \ref{H1}, \ref{H2}, \ref{ass1-explicit}, \ref{ass4-explicit}, and \ref{ass5-explicit}, there are a positive smooth function $\rho (\cdot)$ and {\myr a large enough positive number} $k$ such that the regulator
\begin{subequations}\label{explicit-mas}
\begin{align}
\dot{\eta} &=M\eta+Nu\label{explicit-mas1}\\
u &=-k\rho(e)e + \chi(\eta, \check{a} (\eta )) \label{explicit-mas3}
\end{align}\end{subequations}
solves Problem \ref{ldlesp} with 
\begin{align*}
\chi (\eta, \check{a} (\eta ))&=\Gamma \Xi (\check{a} (\eta))\textnormal{\col}(\eta_{1},\dots,\eta_{n}),\\
\check{a} (\eta)&= -\Theta (\eta)^{-1}\textnormal{\col}(\eta_{n +1},\dots,\eta_{2n}),
\end{align*}
provided that 
$\Theta ^{-1}(\eta(t))$ always exists and is well-defined for any $t\geq 0$. 
\end{thm}
The results of Theorem \ref{cor-Lemma-explicit} can be obtained directly from Theorem \ref{Theorem-1}.

The existence of the nonlinear mapping $\chi \left(\eta, \check{a} (\eta )\right)$ strictly relies on the solution of a time-varying equation $$\Theta(\eta)\check{a} (\eta )+\textnormal{\col}(\eta_{n +1},\dots,\eta_{2n})=0.$$ It is noted that $\Theta(\eta(t))$ is not always invertible over $t\geq 0$ as the inverse of $ \Theta(\eta)$ may not exist or may not even be well-defined.
To achieve Objective \ref{obj2}, we introduce a nonparametric learning framework to solve the time-varying linear equation arising from the signal of the steady-state generator. The proposed framework differs from the classical parametric adaptive control approaches investigated in  \cite{liu2009parameter,serrani2001semi,adetola2014adaptive,tomei2023adaptive,yucelen2012low,broucke2024disturbance}, {\myr which relies on the explicit regressor construction typically generated from the structure of the composite system (consisting of the controlled system and internal model). The proposed framework does not require this explicit construction (see Fig.~\ref{framework}).}

\subsection{A nonparametric learning framework for the nonlinear robust output regulation problem}
\begin{figure}[htp]
    \centering
\includegraphics[width=0.48\textwidth,trim=40 0 0 0]{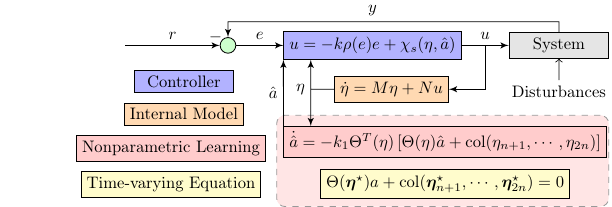}  
\caption{A nonparametric learning solution to the time-varying linear equation.}\label{framework}
\end{figure}

From  Assumption \ref{ass1-explicit}, it follows that $\bm{\eta} ^{\star}$ and $a$ belong to some compact set $\mathds{D}$. {\myr 
Given that the signal 
$\theta $ is inaccessible, and the constant vector 
$a$ cannot be directly calculated through \eqref{a-theta}, we must rely solely on the internal model 
$\eta$ for the composite system \eqref{Main-sys1}. Consequently, we propose the following nonparametric learning regulator. This approach circumvents the need for direct access to 
$\theta$ and allows for the estimation of 
$a$ using the available internal model signals, thus enhancing the robustness and adaptability of the control system in scenarios where direct measurement of 
$\theta$ is not feasible.
As a result, for the composite system \eqref{Main-sys1}, we propose the  nonparametric learning regulator}
\begin{subequations}\label{explicit-mas}
\begin{align}
\dot{\eta} &=M\eta+Nu\label{explicit-mas1}\\
\dot{\hat{a}}  &=- k_1 \Theta(\eta )\!^\top\!\left[\Theta (\eta )\hat{a} +\textnormal{\col}(\eta_{n +1},\dots,\eta_{2n})\right] \\
u &= -k \rho(e)e + \chi_s(\eta,\hat{a}) \label{explicit-mas3}
\end{align}\end{subequations}
where $\hat{a}$ is the estimation of unknown parameter vector $a$, $\rho (\cdot)\geq 1$ is a positive smooth function, and the smooth function {\myr $\chi_s(\eta,\hat{a})$ is given by
\begin{align*}
\chi_s(\eta, \hat{a})=\chi(\eta, \hat{a})\Psi(\delta+1-\|\col(\eta, \hat{a})\|^2),
\end{align*}
having a compact support with
$$\chi(\eta,\hat{a})\,{\myr \equiv}\,\Gamma \Xi (\hat{a} )\textnormal{\col}(\eta_{1},\dots,\eta_{n}),$$
 and $\Psi(\varsigma)=\frac{\psi(\varsigma)}{\psi(\varsigma)+\psi(1-\varsigma)}$, 
$\delta=\max_{(\eta,\hat{a})\in \mathds{D}}\|\col(\eta, \hat{a})\|^2$, and $$\psi(\varsigma)=\left\{\begin{array}{@{}cc}e^{-1/\varsigma} & \textnormal{for}\;\varsigma >0,\\
0 & \textnormal{for}\;\varsigma \leq0.
\end{array}\right.$$}
\begin{rem} 
One of the major differences with the adaptive methods in \cite{nikiforov1998adaptive,serrani2001semi,chen2002global,liu2009parameter} is that the nonparametric framework proposed in Fig.\ \ref{framework} does not rely on the regressor typically generated from the structure of the composite system, which consists of the controlled system and internal model. Instead, the nonparametric framework directly learns the unknown parameter vector 
$a$ from the internal model signal $\eta$ as shown in Fig.~\ref{framework}.
This approach eliminates the need for explicit regressor construction, potentially simplifying implementation and reducing computational complexity. Additionally, it allows for more flexible adaptation to various system dynamics without predefined model assumptions.
\end{rem}
 
We now perform the coordinate/input transformations
\begin{align*}
\bar{\eta} =&\,\eta-\bm{\eta}^{\star}-\bar{b}(\bar{z},e,\mu)^{-1}Ne,\;\\
\bar{u}_s= &\,u-\chi_s(\eta,\hat{a})\;\;\textnormal{and}\;\;\bar{a}=\hat{a}-a,
\end{align*}
which leads to the augmented system: 
\begin{subequations}\label{bara-deriv}\begin{align}
\dot{\bar{z}} =&\,\bar{f}(\bar{z},e,\mu), \label{bara-1a}\\
 \dot{\bar{\eta}} =&\,M\bar{\eta}+\bar{p}(\bar{z},e,\mu),\label{bara-1b}\\
\dot{e} =&  \,\bar{g}(\bar{z},\bar{\eta},e,\mu)  + \bar{b}(\bar{z},e,\mu)\bar{\chi}_s(\bar{\eta}_e,\bar{a})+ b(\bar{z},e,\mu)\bar{u}_s,\label{generic-1c}\\
\dot{\bar{a}}  =&-k_1 \Theta (\bm{\eta}^{\star} )\!^\top\!\Theta (\bm{\eta}^{\star} )\bar{a} - k_1 \bar{O}(\bar{\eta}_e,\bar{a}),
\end{align} \end{subequations}
where $\bar{f}(\bar{z},e,\mu)$, $\bar{p}(\bar{z},e,\mu)$, and $\bar{g}(\bar{Z},e,\mu)$ are defined in \eqref{augmen-1}, and
\begin{align*}
\bar{\chi}_s(\bar{\eta}_e,\bar{a},\mu)=&\, \chi_s\left(\bm{\eta}^{\star}+\bar{\eta}_e,\bar{a}+a\right)-\chi( \bm{\eta}^{\star},a)\\
\bar{\eta}_e=&\,\bar{\eta}+\bar{b}(\bar{z},e,\mu)^{-1}Ne\\  
\bar{O}(\bar{\eta}_e,\bar{a})=&\,\Theta (\bm{\eta}^{\star} )\!^\top\!\Theta (\bar{\eta}_e  )\bar{a}+\Theta (\bar{\eta}_e  )\!^\top\!\Theta (\bm{\eta}^{\star} )\bar{a}\\
&+\Theta(\bar{\eta}_e ) \! ^\top \! \Theta (\bar{\eta}_e  )a+\Theta (\bar{\eta}_e  )\!^\top\!\Theta (\bar{\eta}_e )\bar{a}\\
&+\Theta (\bm{\eta}^{\star} )\!^\top\!\Theta (\bar{\eta}_e  ){a} +\Theta (\bm{\eta}^{\star} )\!^\top\!\textnormal{\col}(\bar{\eta}_{e,n +1},\dots,\bar{\eta}_{e,2n})\\
&+\Theta (\bar{\eta}_e  )\!^\top\!\textnormal{\col}(\bar{\eta}_{e,n +1},\dots,\bar{\eta}_{e,2n}).
\end{align*}
Then, we can introduce the following properties of the control law \eqref{explicit-mas1} under Assumptions \ref{H1}, \ref{ass1-explicit}, \ref{ass4-explicit}, and \ref{ass5-explicit} {\myr with the following additional assumption on the function $b(z,y,\mu)$ that it is upper bounded by a positive constant $\bar{b}^*$ for any $\mu\in \mathds{V}\times \mathds{W} \times \mathds{S}$ and $\col(z,y)\in\mathds{R}^{n_z+1}$.}
\begin{lem}\label{lemmabodev} %
For the system \eqref{bara-deriv} under Assumptions \ref{H1}, \ref{H2}, \ref{ass1-explicit}, \ref{ass4-explicit}, and \ref{ass5-explicit}, we have the following properties:
\begin{prop}\label{property3}%
There are smooth integral Input-to-State Stable Lyapunov functions $V_{\bar{a}}{\myr \equiv}V_{\bar{a}}\big(\bar{a}\big)$ satisfying
\begin{align}
\underline{\alpha}_{\bar{a}}(\|\bar{a}\|^2)&\leq V_{\bar{a}}(\bar{a})\leq \bar{\alpha}_{\bar{a}}(\|\bar{a}\|^2),\notag\\
\dot{V}_{\bar{a}}\big|_{\eqref{explicit-mas1}} &\leq  -\alpha_{\bar{a}}(V_{\bar{a}}) +c_{ae}\|\bar{Z}\|^2+c_{ae}e^2,\label{V2-tildea}
\end{align}
for positive constant $c_{ae}$, and comparison functions $\underline{\alpha}_{\bar{a}}(\cdot)\in {\myr \mathcal{K}_{\infty}}$, $\bar{\alpha}_{\bar{a}}(\cdot)\in {\myr \mathcal{K}_{\infty}}$, {\myr $\alpha_{\bar{a}}(\cdot)\in \mathcal{K}_{o}$}.
\end{prop}
\begin{prop} \label{property4} There are positive constants $\phi_0$, $\phi_1$, and $\phi_2$ such that
\begin{align*}
{\myr|\bar{b}(\bar{z},e,\mu)|^2|\bar{\chi}_s(\bar{\eta}_e,\bar{a},\mu)|^2}\leq\ &\phi_0e^2+\phi_1\|\bar{Z}\|^2+\phi_{2}\alpha_{\bar{a}}(V_{\bar{a}})
\end{align*}
{\myr for $\mu\in \mathds{V}\times \mathds{W}\times \mathds{S}$.}
\end{prop}
\end{lem}
\begin{proof} 
From Lemma \ref{theta-invertable-explicit}, we have $\Theta (\bm{\eta}^{\star} )$ is non-singular and $\bm{\eta}^{\star}(t)$ is {\myr bounded} signal for all $t\geq 0$. Therefore, there exists positive constants $\Theta_m$ and $\Theta_M$ such that $$\Theta_m I_n \leq \Theta(\bm{\eta}^{\star}(t) )\!^\top\!\Theta (\bm{\eta}^{\star}(t) )\leq \Theta_M I_n,\;\;\forall t\geq 0.$$
It is noted that $\|\bm{\eta}^{\star}(t)\|\leq \eta_M$ over Assumption \ref{ass1-explicit} for some positive constant $\eta_M$ and for all $t\geq 0$. Together with $\|\Theta(\bar{\eta}_e )\|\leq n^{1/2}\|\bar{\eta}_e\|$ and $\bar{a}^\top\Theta(\bar{\eta}_e ) \! ^\top \! \Theta (\bar{\eta}_e )\bar{a}\geq 0$, we have
\begin{align*}
-\bar{a}^\top\bar{O}(\bar{\eta}_e,\bar{a}) \leq& -\!2\bar{a}^\top\Theta (\bm{\eta}^{\star} )\!^\top\!\Theta (\bar{\eta}_e )\bar{a} -\bar{a}^\top\Theta(\bar{\eta}_e ) \! ^\top \! \Theta (\bar{\eta}_e )a\\
& -\!\bar{a}^\top\Theta (\bm{\eta}^{\star} )\!^\top\!\left[\Theta (\bar{\eta}_e ){a}+\textnormal{\col}(\bar{\eta}_{e,n +1},\dots,\bar{\eta}_{e,2n})\right]\\
& -\!\bar{a}^\top\Theta(\bar{\eta}_e ) \! ^\top \! \textnormal{\col}(\bar{\eta}_{e,n +1},\dots,\bar{\eta}_{e,2n})\\
\leq & \, 2n\|\bar{a}\|\|\bar{\eta}_e\|\left(\eta_M\|\bar{a}\|+\|a\|\eta_M+\|a\|\|\bar{\eta}_e\|\right)\\
\leq &  \,\|a\|^2(\eta_M^2+1)\frac{4n^2}{\Theta_m}\|\bar{\eta}_e\|^2\\
&+\eta_M^2\|\bar{a}\|^2\frac{4n^2}{\Theta_m}\|\bar{\eta}_e\|^2+\tfrac{3\Theta_m}{4}\|\bar{a}\|^2\\
\leq & \, c_{\bar{a}}\|\bar{\eta}_e\|^2\|\bar{a}\|^2+2c_{\bar{a}}\|\bar{\eta}_e\|^2+{\tfrac{3}{4}\Theta_m}\|\bar{a}\|^2
\end{align*}
where $c_{\bar{a}}=\max\{4n^2\eta_M^2,4n^2\|a\|^2\eta_M^2, 4n^2\|a\|^2\}/\Theta_m$. 
Pose the Lyapunov function candidate $V_{\bar{a}}(\bar{a})=\ln(1+\|\bar{a}\|^2)$, which satisfies  $\|\bar{a}\|^2=e^{V_{\bar{a}}}-1$. Then, we have 
$\underline{\alpha}_{\bar{a}}(\varsigma)= \ln(\underline{c}_1 \varsigma^2+1)$ and $\bar{\alpha}_{\bar{a}}(\varsigma)=\bar{c}_1\varsigma^2$ for all $\varsigma>0$ and some positive constants $\underline{c}_1$ and $\bar{c}_1$. {\myr It is noted that $\bar{b}(\bar{z},e,\mu)=b(z,y,\mu)$ is {\myr upper and }lower bounded by some positive constants $\bar{b}^*$ and $b^*$} for any $\mu\in \mathds{V}\times \mathds{W} \times \mathds{S}$ and $\col(z,y)\in\mathds{R}^{n_z+1}$.  
The time
derivative of $V_{\bar{a}}$ along \eqref{bara-deriv} can be evaluated as
\begin{align}\label{V2alpha}
\dot{V}_{\bar{a}}&=\frac{-2 k_1 \bar{a}^\top\Theta (\bm{\eta}^{\star} )\!^\top\!\Theta (\bm{\eta}^{\star} )\bar{a} -2 k_1 \bar{a}^\top\bar{O}(\bar{\eta}_e, \bar{a})}{1+\|\bar{a}\|^2}\nonumber\\
&\leq \frac{-(1/2) k_1 \Theta_m\|\bar{a}\|^2+ k_1 c_{\bar{a}}\|\bar{\eta}_e\|^2\|\bar{a}\|^2+2c_{\bar{a}}\|\bar{\eta}_e\|^2}{1+\|\bar{a}\|^2}\nonumber\\
&\leq \frac{-(1/2) k_1 \Theta_m\|\bar{a}\|^2}{1+\|\bar{a}\|^2}+3 k_1 c_{\bar{a}}\|\bar{\eta}_e\|^2\nonumber \\
&\leq - \alpha_{\bar{a}}(V_{\bar{a}})+c_{ae}(\|\bar{Z}\|^2+ e^2) 
\end{align}
where $\alpha_{\bar{a}}(\varsigma)=\frac{\Theta_m k_1 (e^{\varsigma}-1)}{2e^{\varsigma}}$ and $c_{ae}=\max\{6 k_1 c_{\bar{a}},2/b^*\}$.

We now verify Property \ref{property4}.  
It can be verified that the function $\bar{\chi}_s(\bar{\eta}_e,\bar{a},\mu)=\bar{\chi}_s(\bar{\eta}+\bar{b}^{-1}Ne,\bar{a},\mu)$ is continuous and vanishes at $\col(\bar{z}, e, \bar{\eta})=\col(0,0,0)$.  
From \eqref{explicit-mas}, the function $\chi_s(\eta, \hat{a},\mu)$ is bounded for all $\col(\eta, \hat{a})\in \mathds{R}^{3n}$; by  \cite[Lemma A.1]{xu2016output} and \cite[Remark A.1]{xu2016output}, there exist  $\gamma_1, \gamma_2\in {\myr \mathcal{K}_o}\cap \mathcal{O}(\emph{Id})$ such that
\begin{align*}
|\bar{b}(\bar{z},e,\mu)|^2|\bar{\chi}_s(\bar{\eta}_e,\bar{a},\mu)|^2\leq \gamma_1(\|\bar{\eta}_e\|^2)+\gamma_2(\|\bar{a}\|^2),
\end{align*}
for $\forall\zeta\in \mathds{R}^{3n}$ and $\mu\in \mathds{V}\times \mathds{W}\times \mathds{S}$.
From \eqref{V2alpha}, it can also be verified that $\limsup_{\varsigma\rightarrow 0^{+}}\frac{\gamma_1(\varsigma)}{\varsigma}< +\infty$ and
{\myr \begin{align*}
\limsup_{\varsigma\rightarrow 0^{+}}\frac{\gamma_2\circ [(e^\varsigma-1)\underline{c}_1^{-1}]}{\alpha_{\bar{a}}(\varsigma)}=&\limsup_{\varsigma\rightarrow 0^{+}}\frac{\gamma_2\circ [(e^\varsigma-1)\underline{c}_1^{-1}]}{(e^\varsigma-1)\underline{c}_1^{-1}}\\
&\times\limsup_{\varsigma\rightarrow 0^{+}}\frac{(e^\varsigma-1)\underline{c}_1^{-1}}{\alpha_{\bar{a}}(\varsigma)}< +\infty.\end{align*}}
Hence, there exist positive constants $\phi_1$ and $\phi_2$ such that
\begin{align*}
|\bar{b}(\bar{z},e,\mu)|^2|\bar{\chi}_s(\bar{\eta}_e,\bar{a},\mu)|^2\leq \phi_0e^2+\phi_1\|\bar{Z}\|^2+\phi_{2}\alpha_{\bar{a}}(V_{\bar{a}}),
\end{align*}
for $\mu\in \mathds{V}\times \mathds{W}\times \mathds{S}$. Therefore, we have Property \ref{property4} for system \eqref{bara-deriv}.
\end{proof}

\begin{thm}\label{Thm-Lemma-explicit}%
For the composite system \eqref{Main-sys1} and \eqref{Exosys1} under Assumptions \ref{H1}, \ref{H2}, \ref{ass1-explicit}, \ref{ass4-explicit} and \ref{ass5-explicit}, there exists {\myr a large enough positive constant  $k$} and positive smooth function $\rho (\cdot)$, such that the controller \eqref{explicit-mas}
solves Problem \ref{ldlesp} and achieves Objective \ref{obj2} for any positive constant {\myr $k_1>0$.}
\end{thm}
\begin{proof}
The error dynamics \eqref{augmen-1} with control \eqref{explicit-mas} are given by
\begin{equation}\label{generic-error-1}
\begin{aligned}
 \dot{\bar{z}} &= \bar{f}\left(\bar{z},e,\mu\right), \\
    \dot{\bar{\eta}}&= M\bar{\eta}+\bar{p}(\bar{z},e,\mu),\\
   \dot{\bar{a}} &=- k_1 \Theta (\bm{\eta}^{\star} )\!^\top\!\Theta (\bm{\eta}^{\star} )\bar{a} - k_1 \bar{O}({\myr \bar{\eta}_e},\bar{a}),\\
\dot{e}&=\bar{g}\left(\bar{z},\bar{\eta},e,\mu\right) +\bar{b}(\bar{z},e,\mu)\left(\bar{\chi}_s(\bar{\eta}_e,\bar{a},\mu)-k {\rho}(e)e\right).
\end{aligned}
\end{equation}
From Property \ref{property1}, with the same development and by using the changing supply rate technique \cite{sontag1995changing} again, given any smooth function $\Theta_z (\bar{Z} )>0$, there exists a $C^{1}$ function $V_{2}(\bar{Z})$ satisfying
$$\underline{\alpha}_{2}\big(\big\|\bar{Z} \big\|^2\big)\leq V_{2}\big( \bar{Z}  \big)\leq\overline{\alpha}_{2}\big(\big\|\bar{Z} \big\|^2\big)$$
 for some class $\mathcal{K}_{\infty}$ functions $\underline{\alpha}_{2}\left(\cdot\right)$ and $\overline{\alpha}_{2}\left(\cdot\right)$, such that, for all $\mu\in \mathds{V}\times \mathds{W}\times \mathds{S}$, along the trajectory of the {\myr $\bar{Z} $-subsystem}, $$\dot{V}_{2} \leq-\Theta_z (\bar{Z} )\big\|\bar{Z} \big\|^2+ \hat{\gamma} (e)e^2, $$
where $\hat{\gamma} \left(\cdot\right)\geq 1$ is some known smooth positive definite function. We now define a Lyapunov function $ U{\myr \equiv} U (\bar{Z}, \bar{a}, e)$ by
 \begin{align*}
 U(\bar{Z}, \bar{a}, e)=V_{2}\big( \bar{Z}  \big)+\epsilon_{\bar{a}} V_{\bar{a}}(\bar{a})+e^2
\end{align*}
 where the positive constant $\epsilon_{\bar{a}}$ is to be specified. From Property \ref{property1} and Lemma \ref{lemmabodev}, the time derivative of $U(t)$ along \eqref{generic-error-1} can be evaluated as
\begin{align}
    \dot{U}=&\,\dot{V}_{2}+\epsilon_{\bar{a}} \dot{V}_{\bar{a}}+2e\dot{e}\notag\\
    \leq&-\Theta_z (\bar{Z} )\big\|\bar{Z} \big\|^2 +3\epsilon_{\bar{a}}c_{\bar{a}}\|\bar{Z}\|^2-\epsilon_{\bar{a}}\alpha_{\bar{a}}(V_{\bar{a}})\notag\\
   & +b(\bar{z},e, \mu)^2|\bar{\chi}_s(\bar{\eta}_e,\bar{a},\mu)|^2+|\bar{g}(\bar{z},\bar{\eta},e,\mu)|^2\notag\\
    & -2k\bar{b}(\bar{z},e,\mu) {\rho}(e)e^2+e^2+ \hat{\gamma} (e)e^2\notag\\
    \leq&-(2k\bar{b}(\bar{z},e,\mu) {\rho}(e)-\gamma_{g1}(e)-1-c_{ae}-\phi_0-\hat{\gamma}(e))e^2\notag\\
    &-(\Theta_z (\bar{Z} )-\epsilon_{\bar{a}}c_{ae}-\gamma_{g0}(\bar{Z})-\phi_1)\big\|\bar{Z} \big\|^2 \notag\\
    &-(\epsilon_{\bar{a}}-\phi_2)\alpha_{\bar{a}}(V_{\bar{a}}).\notag
   \end{align}
It is noted that $\bar{b}(\bar{z},e,\mu)=b(z,y,\mu)$ is lower bounded by some positive constant $b^*$ for any $\mu\in \mathds{V}\times \mathds{W} \times \mathds{S}$ and $\col(z,y)\in\mathds{R}^{n_z+1}$.  
Hence, let the smooth functions $\Theta_z(\cdot )$, ${\rho}(\cdot)$, and the positive numbers $\epsilon_{\bar{a}}$ and ${k}$ be such that $\Theta_z (\bar{Z})\geq \epsilon_{\bar{a}}c_{ae}+\gamma_{g0}(\bar{Z})+\phi_1+1$, $\rho(e)\geq \max\{\hat{\gamma} \left(e \right), \gamma_{g1}(e), 1\}$, $\epsilon_{\bar{a}}\geq \phi_2+1$, and {\myr $k> \frac{5}{2b^*}\max\{2,\phi_0,c_{ae}\}\equiv k^*$}. Thus, it follows that 
\begin{align}\label{strlyafun}{\myr \dot{U}}\leq -\big\|\bar{Z}\big\|^2-e^2-\alpha_{\bar{a}}(V_{\bar{a}}).\end{align}
Since the function $U(\tilde{Z}, \bar{a}, {e})$ is positive definite and radially unbounded and satisfies inequality \eqref{strlyafun}, it is a strict Lyapunov function candidate, and,
therefore, it can be concluded that system \eqref{generic-error-1} is uniformly asymptotically stable for all $\col(v,w,\sigma)\in \mathds{V}\times \mathds{W}\times \mathds{S}$.
 This completes the proof.
\end{proof}
Based on the Corollary \ref{Theorem-2} and Theorem \ref{Thm-Lemma-explicit}, we also have the following result.
\begin{cor}\label{cor-Lemma-explicit}%
For the composite system \eqref{Main-sys1} and \eqref{Exosys1} under Assumptions \ref{H1} and \ref{H2}, there is a positive smooth function ${\rho}(\cdot)$, and for any positive constant $k_1$ such that the controller
\begin{subequations}\label{cor-ESC-2}%
\begin{align}
\dot{\eta} &=M\eta+Nu\label{cor-explicit-mas1}\\
\dot{\hat{a}}  &=- k_1 \Theta(\eta )\!^\top\!\left[\Theta (\eta )\hat{a} +\textnormal{\col}(\eta_{n +1},\dots,\eta_{2n})\right] \\
\dot{\hat{k}}&={\rho}(e)e^2\\
u &=-\hat{k} \rho(e)e + \chi_s(\eta,\hat{a}) \label{cor-explicit-mas3}
\end{align}\end{subequations}
solves Problem \ref{ldlesp}. 
\end{cor}%


\section{A Nonparametric Learning Framework for Parameter Estimation and Feedforward Control}\label{feedforwadsection}
\subsection{A single-input single-output (SISO) linear system}
We now apply our nonparametric learning framework in feedforward control design to solve the linear output regulation for the SISO system (see \cite{marino2007output}):
\begin{subequations}\label{nonparametric-feedforward}
\begin{align}
\dot{x}&=Ax+Bu+Pv,\quad x\in \mathds{R}^{n_x},\\
e&=Cx+Du+Fv,\quad e\in \mathds{R},\; u\in \mathds{R},
\end{align}    
\end{subequations}
where $A$, $B$, $P$, $C$, $D$, and $F$ are matrices of compatible dimensions, the objective is to design a feedforward control system that drives the tracking error $e$ to zero and rejects the disturbance. The exosystem state variable, $v\in\mathds{R}^{n}$, is generated by \eqref{Exosys1} with 
\[\myr s(v,\sigma)=\Phi (a(\sigma))v,\;y_0=\Gamma v\]
satisfying Assumption \ref{ass1-explicit} with output $y_0$. 
Under Assumption \ref{ass1-explicit}, $y_0$ is a multi-tone sinusoidal signal with $m$ distinct frequencies described as
\[\myr y_0= d_0 +\sum\nolimits_{k=1}^{m}d_{k}(\sin({\omega}_{k}(\sigma)t+\phi_{k}))
\]
where {\myr $d_{0}\in \mathds{R}$,} $d_{k}\neq 0$, $\psi_{k}\in \mathds{R}$, and {\myr ${\omega}_{k}(\sigma)\in \mathds{R}$} are unknown constant biases, amplitudes, initial phases, and frequencies. Moreover, system \eqref{Exosys1} has the characteristic polynomial
\begin{align*}
P(\varsigma)=\varsigma^{n}+a_{n}(\sigma)\varsigma^{n-1}+\dots+a_2(\sigma)\varsigma+a_1(\sigma)
\end{align*}
with {\myr unknown parameter vector $a(\sigma)=\col(a_1(\sigma),\dots, a_{n}(\sigma))$ acting as a reparametrization of the $m$ unknown frequency vector ${\omega}(\sigma)=\col({\omega}_1(\sigma),\dots,{\omega}_{m}(\sigma))$ with $n=2m$ ($d_0=0$) or $n=2m+1$ ($d_0\neq0$) (see \cite{marino2002global}).}


\subsection{Parameter estimation problem of the 
multi-tone sinusoidal signals with unknown frequencies}
\begin{figure}[htp]
    \centering
\includegraphics[width=0.48\textwidth,trim=40 0 0 0]{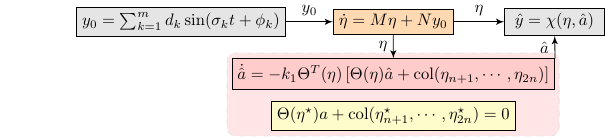}  
\caption{Nonparametric learning framework to parameter/frequency estimation.}\label{parameterframework}
\end{figure}
Following the above results, we can design the observer
\begin{subequations}\label{parameter}\begin{align}
\dot{\eta}&=M\eta+Ny_0\label{parameter1}\\
\dot{\hat{a}}&=-k_1\Theta(\eta)\!^\top\!\left[\Theta(\eta)\hat{a}+\textnormal{\col}(\eta_{n +1},\dots,\eta_{2n})\right]\label{parameter2}
\end{align}\end{subequations}
to reconstruct $y_0$ and estimate the frequencies $\omega_{k}(\sigma)$, 
where $\eta\in \mathds{R}^{2n}$, $M$, and $N$ are chosen based on \eqref{MNINter}. 
\begin{thm}\label{observerMfrequency}For any positive constant $k_1$,
system \eqref{parameter} can globally and exponentially reconstruct and estimate the signals $y_0$ and $\hat{a}$, provided that $y_0$ {\myr has} $m$ distinct frequencies with 
\begin{align*}
    \hat{y}&=\Gamma \Xi (\hat{a} )\textnormal{\col}(\eta_{1},\dots,\eta_{n})\,{\myr \equiv}\,\chi(\eta,\hat{a} ).
\end{align*}

\end{thm}
\begin{proof} Define the coordinate transformation
$$ \tilde{\eta}=\eta-Qv,$$
where $Q$ is given in \eqref{Qdefini}. 
Then, from \eqref{parameter}, the time derivative of $\tilde{\eta}$ along \eqref{parameter1} is given by
 \begin{align*}
\dot{\tilde{\eta}}&=M\eta+Ny_0-Q\Phi(a) v\\
 &=M\tilde{\eta}+MQv+N \Gamma v-Q\Phi(a) v\\
 &=M\tilde{\eta}+(MQ+N \Gamma-Q\Phi(a))v
 \end{align*}
 Then, from generalized Sylvester equation \eqref{MNGAMMAPhi}, we have
 \begin{align*} \dot{\tilde{\eta}}&=M\tilde{\eta},
  \end{align*}
  which further implies that $\tilde{\eta}$ converge to zero exponentially and $\|\eta\|$ is bounded by some positive constant $\eta_M$. We first show that $\hat{a}$ is bounded by some polynomial. For convenience, we let $V_{\hat{a}}=\hat{a}^\top\hat{a}$ and $ k_1 =1$. The time derivative along the trajectory \eqref{parameter2} can be evaluated as
\begin{align*}
\dot{V}_{\hat{a}}&=2\hat{a}^\top\dot{\hat{a}}\\
&=-2\hat{a}^\top\Theta(\eta)\!^\top\!\left[\Theta(\eta)\hat{a}+\textnormal{\col}(\eta_{n +1},\dots,\eta_{2n})\right]\\
&\leq 2\hat{a}^\top\Theta(\eta)\!^\top\!\textnormal{\col}(\eta_{n +1},\dots,\eta_{2n})\\
&\leq 2n^{1/2}\eta_M\sqrt{V_{\hat{a}}},
   \end{align*}
   which further implies that $\hat{a}$ is bounded by some polynomial $q_{\hat{a}}(t)$.
  Now, let $\bar{a}=\hat{a}-a$ and ${\eta}^{\star}=Qv$; from \eqref{parameter}, we have 
 \begin{align*}
  \dot{\bar{a}}=&-\Theta({\eta}^{\star})\!^\top\!\Theta({\eta}^{\star})\bar{a}\\
  &-\Theta({\eta}^{\star})\!^\top\!\left[\Theta(\tilde{\eta})\hat{a}+\textnormal{\col}(\tilde{\eta}_{n +1},\dots,\tilde{\eta}_{2n})\right]\\
  &-\Theta(\tilde{\eta})\!^\top\!\left[\Theta(\eta)\hat{a}+\textnormal{\col}(\eta_{n +1},\dots,\eta_{2n})\right],
    \end{align*}
 Let $V_{\bar{a}}=\bar{a}^\top\bar{a}$. The time derivative of $V_{\bar{a}}$ can be evaluated as 
  \begin{align}\label{dotVbara}
\dot{V}_{\bar{a}}=& - 2\bar{a}^\top\Theta({\eta}^{\star})\!^\top\!\Theta({\eta}^{\star})\bar{a}\notag\\
&- 2\bar{a}^\top\Theta({\eta}^{\star})\!^\top\!\left[\Theta(\tilde{\eta})\hat{a}+\textnormal{\col}(\tilde{\eta}_{n +1},\dots,\tilde{\eta}_{2n})\right]\notag\\
  &- 2\bar{a}^\top\Theta(\tilde{\eta})\!^\top\!\left[\Theta(\eta)\hat{a}+\textnormal{\col}(\eta_{n +1},\dots,\eta_{2n})\right].
  \end{align}
Since the signal $y_0$ is sufficiently
rich of order $2m$, the vector $v$ is persistently exciting. From Lemma \ref{theta-invertable-explicit}, we have $\Theta ({\eta}^{\star}(t) )$ is non-singular and ${\eta}^{\star}(t)$ is {\myr bounded by some positive constant $\eta_M$ for all $t\geq 0$}. Therefore, there exists positive constants $\Theta_m$ and $\Theta_M$ such that $$\Theta_m I_n \leq \Theta({\eta}^{\star}(t) )\! ^\top\! \Theta ({\eta}^{\star}(t))\leq \Theta_M I_n,\;\;\;\;\forall t\geq 0.$$
Then there exists some constant $v_m$ such that \eqref{dotVbara} satisfies
  \begin{align*}
\dot{V}_{\bar{a}}\leq& - 2V_{\bar{a}}+v_m(1+\|\hat{a}\|)\|\tilde{\eta}\|.
  \end{align*}
  We have proved that $\tilde{\eta}(t)$ converges to zero exponentially as $t\rightarrow \infty$ and $\|\hat{a}(t)\|$ is bounded by some polynomial $q_{\hat{a}}(t)$ for all $t\geq 0$. Hence, $v_m(1+\|\hat{a}(t)\|)\|\tilde{\eta}(t)\|$ converges to zero exponentially as $t\rightarrow \infty$, which further implies that $V_{\bar{a}}(t)$ converges to zero exponentially as $t\rightarrow \infty$.
It is noted from \eqref{xi-solution-ex} that 
\begin{align*}
    |\hat{y}-y_0|&=|\chi(\eta,\hat{a} )-\Gamma v|\\
    &=|\chi(\eta,\hat{a} )-\Gamma \Xi (a)\col(\eta_{1}^*,\dots,\eta_{n}^*)|\\
    &\leq\|\Gamma\|\left\|\left(\Xi (\bar{{a}} +a)-\Xi (a)\right)\textnormal{\col}(\eta_{1},\dots,\eta_{n})\right.\\
    &\quad -\Xi (a)\col(\tilde{\eta}_{1},\dots,\tilde{\eta}_{n})\|\\
    &\leq\|\Gamma\|\left(\|\eta\|\|\Xi (\bar{{a}} +a )-\Xi (a)\|+\|\Xi (a)\|\|\tilde{\eta}\|\right).
\end{align*}
which implies that $\lim\limits_{t\rightarrow \infty}\left(\hat{y}(t)-y_0(t)\right)=0$ exponentially. 
\end{proof}
Then, we have the following result.
\begin{cor}\label{corra-parameter}\myr $\Theta(\eta(t))$ is an invertible matrix-valued smooth function for $t\geq t_0$ and 
$$\lim\limits_{t\rightarrow \infty}[a-\Theta(\eta(t))^{-1}\textnormal{\col}(\eta_{n +1}(t),\dots,\eta_{2n}(t))]=0,\;\; \forall t\geq t_0$$
exponentially.
\end{cor}
\begin{proof} From Lemma \ref{theta-invertable-explicit}, we have $\Theta({\eta}^{\star}(t))$ is non-singular for any $t\geq 0$ with ${\eta}^{\star}(t)=Qv(t)$. 
Besides ${\eta}^{\star}(t)$ is {\myr bounded} for all $t\geq 0$, there exists a constant $\epsilon_{\Theta}$ such that $$ \Theta({\eta}^{\star}(t))\!^\top\!\Theta({\eta}^{\star}(t)) \geq \epsilon_{\Theta}I_n,\;\;\;\;\forall t\geq 0.$$
From Theorem \ref{observerMfrequency}, $\eta(t)$ and ${\eta}^{\star}(t)$ are uniformly bounded for all $t\geq0$, and $\lim\limits_{t\rightarrow \infty}\left(\eta(t)-{\eta}^{\star}(t)\right)=0$ exponentially, which implies that
\begin{align*}
\lim\limits_{t\rightarrow \infty}\left[\Theta({\eta}^{\star}(t))\!^\top\!\Theta({\eta}^{\star}(t))-\Theta(\eta(t))\!^\top\!\Theta(\eta(t))\right]=0
\end{align*}
exponentially. Hence there exists a time $t_0$ such that, for $t\geq t_0$,
\begin{align*}
-\frac{\epsilon_{\Theta}}{2}I_n \leq \left[\Theta(\eta(t))\!^\top\!\Theta(\eta(t))-\Theta({\eta}^{\star}(t))\!^\top\!\Theta({\eta}^{\star}(t))\right]\leq \frac{\epsilon_{\Theta}}{2}I_n.
\end{align*}
Then, for $t\geq t_0$, we have that
\begin{align*}
\frac{\epsilon_{\Theta}}{2}I_n \leq \Theta(\eta(t))\!^\top\!\Theta(\eta(t))\leq \frac{3\epsilon_{\Theta}}{2}I_n,
\end{align*}
which implies that $\Theta(\eta(t))$ is an invertible matrix-valued smooth function for $t\geq t_0$.  Let $a_\eta(t)=-\Theta(\eta(t))^{-1}\textnormal{\col}(\eta_{n +1}(t),\dots,\eta_{2n}(t))$ for $t\geq t_0$. From Theorem \ref{observerMfrequency}, $\eta(t)$ and ${\eta}^{\star}(t)$ are uniformly bounded for $t\geq 0$. 
Then, for some constant $c_{\Theta}$, from \eqref{a-explicit0}, we have that
\begin{align*}
\|a-a_\eta\|&=\left\|\Theta({\eta}^{\star})^{-1}\textnormal{\col}(\eta_{n +1}^*,\dots,\eta_{2n}^*)\right.\\
&\quad\quad -\Theta(\eta)^{-1}\textnormal{\col}(\eta_{n +1},\dots,\eta_{2n})\|\\
&=\left\|\left[\Theta({\eta}^{\star})^{-1}-\Theta(\eta)^{-1}\right]\textnormal{\col}(\eta_{n +1}^*,\dots,\eta_{2n}^*)\right.\\
&\quad\quad-\Theta(\eta)^{-1}\textnormal{\col}(\tilde{\eta}_{n +1},\dots,\tilde{\eta}_{2n})\|\\
&\leq  c_{\Theta}(\left\|\Theta({\eta}^{\star})^{-1}-\Theta(\tilde{\eta}+{\eta}^{\star})^{-1}\right\|+\|\tilde{\eta}\|).
\end{align*}
{\myr It can be verified that the function $\left\|\Theta({\eta}^{\star}(t))^{-1}-\Theta(\tilde{\eta}(t)+{\eta}^{\star}(t))^{-1}\right\|$ is { continuous} for $t\geq t_0$ and vanishes at $\tilde{\eta}=0$.
 By \cite[Lemma 11.1]{chen2015stabilization}, there exists a positive smooth function $\gamma_{\Theta}(\cdot)$ such that
$$\left\|\Theta({\eta}^{\star})^{-1}-\Theta(\tilde{\eta}+{\eta}^{\star})^{-1}\right\|\leq \gamma_{\Theta}(\tilde{\eta}, {\eta}^{\star})\|\tilde{\eta}\|.$$
 As a result, we have that 
 \begin{align*}
 \|a-a_\eta(t)\|&\leq  c_{\Theta}\left(\gamma_{\Theta}(\tilde{\eta}(t), {\eta}^{\star}(t))+1\right)\|\tilde{\eta}(t)\|,\;\;\;\;\forall t\geq t_0.
 \end{align*}
This inequality, together with $\lim\limits_{t\rightarrow \infty}\tilde{\eta}(t)=0$ exponentially, implies that $\lim\limits_{t\rightarrow \infty}(a-a_\eta(t))=0$ exponentially for $t\geq t_0$. }
\end{proof}

{\begin{rem}
Many observer design techniques exist for online parameter/frequency estimation using on traditional adaptive control methods as in \cite{marino2002global,xia2002global,hou2011parameter,pin2019identification}.
Unlike the adaptive iterative learning of the learning to minimize error signals, such as in the MIT rule, the non-adaptive control approach naturally avoids the burst phenomenon \cite{anderson1985adaptive}.
In addition, the non-adaptive method in \cite{marconi2007output,isidori2012robust,xu2019generic} avoids the need for the construction of a Lyapunov function or an input-to-state stability (ISS) Lyapunov function for the closed-loop system with positive semidefinite derivatives.
%
In fact, a counterexample was used in \cite{chen2023lasalle} to show that the boundedness property cannot be guaranteed when the derivative of an ISS Lyapunov function is negative semidefinite, even with a small input term. 
Moreover, Corollary \ref{corra-parameter} also demonstrates that after a certain time, the observer \eqref{parameter} can directly provide the estimated parameters.
\end{rem}}
\subsection{Feedforward control design for output regulation}
To solve the output regulation problem by using feedforward control design, we assume that $(A, B)$ is stabilizable and the linear matrix equations:
\begin{align}\label{sim-regulator-equation}
    X \Phi (a)=AX+BU+P\;\textnormal{and}\; 0=CX+DU+F    
\end{align}
admits a solution pair $(X,U)$. As shown in \cite{huang2004nonlinear}, the matrix equations \eqref{sim-regulator-equation} can be put into the form 
$$\mathcal{A}\zeta_{a}=\mathcal{B}$$
where $\zeta_{a}=\textnormal{vec}(\col(X,U))$, $\mathcal{B}=\textnormal{vec}(\col(P,F))$, and 
$$\mathcal{A}=\Phi^\top(a)\otimes \left[\begin{matrix}I_{n_x}&0\\0& 0\end{matrix}\right]-I_n\otimes \left[\begin{matrix}A&B\\C& D\end{matrix}\right].$$
Since $\Phi(a)$ contains the unknown parameters, $a$, motivated by \cite{wang2021cooperative}, we introduce the equation
\begin{equation} \label{time-varying-estimator}
\dot{\hat{\zeta}}_a = -k_2 \mathcal{A}^\top\!(t) \bigl(\mathcal{A}(t) \hat{\zeta}_{a} - \mathcal{B} \bigr)
\end{equation}
to solve the algebraic equation adaptively, where
$$\mathcal{A}(t)=\Phi^\top\!(\hat{a}(t))\otimes \left[\begin{matrix}I_{n_x}&0\\0& 0\end{matrix}\right]-I_n\otimes \left[\begin{matrix}A&B\\C& D\end{matrix}\right],$$
with $\hat{a}$ is generated by \eqref{parameter2}.
From Lemma 4.1 in \cite{wang2021cooperative}, for any initial condition $\hat{\zeta}_{a}(t_0)$ and any positive constant $k_2$, \eqref{time-varying-estimator}
 adaptively solves the matrix equations \eqref{sim-regulator-equation} in the sense that 
\begin{align}\label{XUCOL}
 \lim\limits_{t\rightarrow\infty} \left(
\textnormal{\col}(\hat{X}(t),\hat{U}(t))
- \textnormal{\col}(X,U)\right)  = 0,
 \end{align}
exponentially, where $\col(\hat{X}(t),\hat{U}(t)) = M_{n_x +1} ^{n}(\hat{\zeta}_{a}(t))$.

We now introduce the feedforward control law for system \eqref{nonparametric-feedforward},
\begin{align}\label{feedforward-contol}
u=K_xx+ \big(\hat{U}(t)-K_x\hat{X}(t)\big) \Xi (\hat{a} )\textnormal{\col}(\eta_{1},\dots,\eta_{n}),
  \end{align}
where $K_x$ is such that $A+BK_x$ is Hurwitz.
\begin{thm}
    For systems \eqref{Exosys1} and \eqref{nonparametric-feedforward}, under Assumption \ref{ass1-explicit},
   for any positive $k_1$ and $k_2$, the output regulation can be solved by the feedforward control law composed of \eqref{parameter}, \eqref{time-varying-estimator}, and \eqref{feedforward-contol}.
\end{thm}
\begin{proof} Let $\bar{x}=x-Xv$, $\bar{u}=u-Uv$, $\bar{X}(t)=\hat{X}(t)-X$, and $\bar{U}(t)=\hat{U}(t)-U$. Then, from \eqref{xi-solution-ex}, we have the equations,
\begin{align}
    \dot{\bar{x}}&= (A+BK_x)\bar{x}+Bd_x(t),\notag\\
   e &=(C+DK_x)\bar{x}+Dd_x(t),\label{baruinput}
\end{align}
where 
\begin{align*}d_x(t)=&-\big(\hat{U}(t)-K_x\hat{X}(t)\big)\Xi (a )\textnormal{\col}(\tilde{\eta}_{1}(t),\dots,\tilde{\eta}_{n}(t)) \\
&+\bar{U}(t)v(t)-K_x\bar{X}v(t)+\big(\hat{U}(t)-K_x\hat{X}(t)\big)\left(-\Xi (a)\right.\\
&\left.+\, \Xi (\bar{a}(t)+a)\right)\textnormal{\col}(\eta_{1}(t),\dots,\eta_{n}(t)).\end{align*}
By Theorem \ref{observerMfrequency},  $\lim\limits_{t\rightarrow\infty}\bar{a}(t)=0$ and $\lim\limits_{t\rightarrow\infty}\tilde{\eta}(t)=0$ exponentially. It is noted from \eqref{XUCOL} that $\lim\limits_{t\rightarrow\infty}
\textnormal{\col}(\bar{X}(t),\bar{U}(t))  = 0$ exponentially. Clearly, $d_x(t)$ converges to the origin exponentially. With $A+BK_x$ taken as a Hurwitz matrix, it can be concluded from  \eqref{baruinput} that $\lim\limits_{t\rightarrow\infty}\bar{x}(t)=0$. Therefore, $\lim\limits_{t\rightarrow\infty}e(t)=0$.
\end{proof}




\section{Numerical and Practical Examples}\label{NumPraEx}
\subsection{Example 1: Controlled Lorenz systems}
\begin{figure}[ht]
 \epsfig{figure=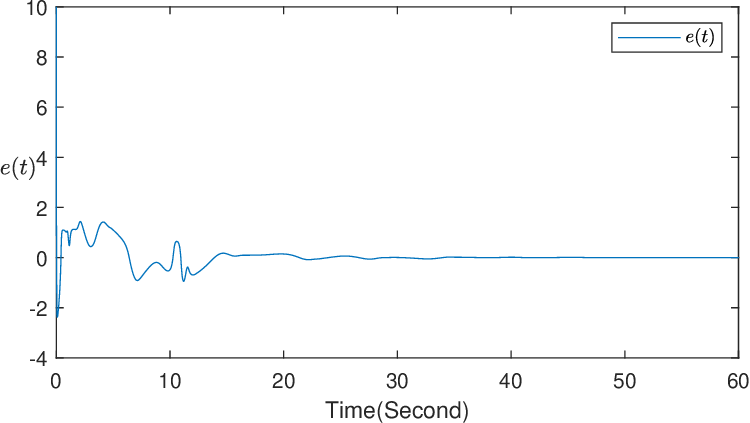,width=0.48\textwidth}
 \caption{Tracking error subject to the controller \eqref{cor-ESC-2}}\label{fige}
\end{figure} 

\begin{figure}[ht]
 \epsfig{figure=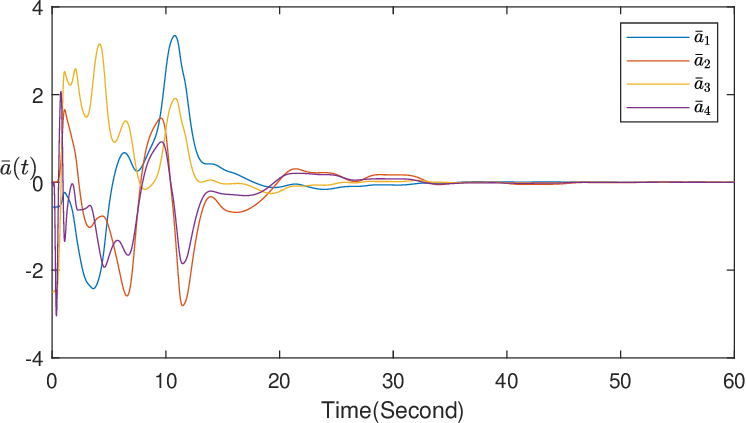,width=0.48\textwidth}
\caption{Parameter estimation error $\bar{a}$ subject to the controller \eqref{cor-ESC-2}}\label{figy}
\end{figure} 
We consider controlled Lorenz systems described by 
\begin{align}\label{numerica1}
\dot{z}=&\left[
                \begin{array}{cc}
                  -L_{1} & 0\\
                  y& -L_{2} \\
                \end{array}
              \right]\!z +\left[
                           \begin{array}{c}
                             L_{1}y \\
                             0 \\
                           \end{array}
                         \right],\nonumber\\
 \dot{y}=&\left[1, 0\right]z \left(L_{3}-\left[0, 1\right]z \right)-y +b(v,w)u,  \nonumber\\
 e =&y-q(v),
\end{align}
where $z=\col(z_{1},z_{2})$ and $y$ are the state, $L =\textnormal{\col}\left(L_{1},L_{2},L_{3}\right)$ is a constant parameter vector that satisfies $L_{1}>0$, $L_{3}<0$. The term $b(v,w)$ is nonzero and lower bounded by some positive constant $b^*$ for any $\mu\in \mathds{V}\times \mathds{W} \times \mathds{S}$. For convenience, let $L =\bar{L} +w_i $, where $\bar{L} =\col(\bar{L}_{1},\bar{L}_{2},\bar{L}_{3})$ is the true value of $L$ and $w =\col(w_{1},w_{2},w_{3})$ is the uncertain parameter of $L$.
We assume that the uncertainty $w_i \in \mathds{W}_i\subseteq\mathds{ R}^3$. The exosystem is described by system \eqref{Exosys1} with 
\begin{align*}
    s(v,\sigma)&=\begin{bmatrix}
    0 & \sigma\\
    -\sigma & 0
   \end{bmatrix} v,\\ q(v)&=\begin{bmatrix}1 & 0\end{bmatrix}v,
\end{align*}
where $\sigma$ is some unknown parameter. Under Assumptions \ref{H1}, \ref{ass1-explicit}, and \ref{ass4-explicit}, it has been shown in \cite{xu2010robust} that there exists an explicit solution of $\textbf{u}(v,\sigma,w)$, polynomial in $v$, satisfying
\begin{align*}
\frac{d^{4}\textbf{u}}{dt^{4}}+a_1 \textbf{u}+a_2 \frac{d\textbf{u}}{dt}+a_3 \frac{d^{2}\textbf{u}}{dt^{2}}+a_4\frac{d^{3}\textbf{u}}{dt^{3}}=0,
\end{align*}
with unknown true value vector $a=\col(9\sigma^4, 0,10 \sigma^2,0)$ in \eqref{stagerator}. For the control law \eqref{cor-ESC-2}, we can choose $\rho(e)=e^2+1$, $\mu=1$ $m_1=1$, $m_2=5.1503$, $m_3=13.301$,
$m_4=22.2016$, $m_5=25.7518$, $m_6=21.6013$, $m_7=12.8005$, and $m_8=5.2001$. The simulation starts with the initial conditions
$x(0)=\col(3,-1)$, $z(0)=0$, $v(0)=\col(10,2)$, and $\hat{k}(0)=1$. 

Fig.\ \ref{fige} shows the trajectory of the tracking error  $e =y -h(v, w)$ over the controller \eqref{cor-ESC-2}.
Fig.\ \ref{figy} shows the trajectories of parameter estimation error $\bar{a} =\hat{a} -a$ for the controller \eqref{cor-ESC-2}. The results demonstrate the effectiveness of the proposed approach.

\subsection{Example 2: Application to quarter-car active automotive suspension systems subject to unknown road profiles}

 \begin{figure}[htbp]\pagestyle{empty}
 \centering
\includegraphics[width=0.48\textwidth]{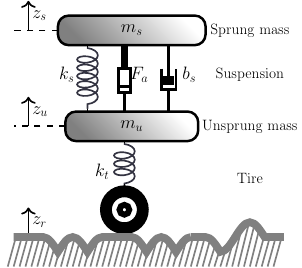}  
\caption{Quarter-car active automotive suspension system}
\label{QCAASS}
\end{figure}

\begin{figure}[ht]
 \epsfig{figure=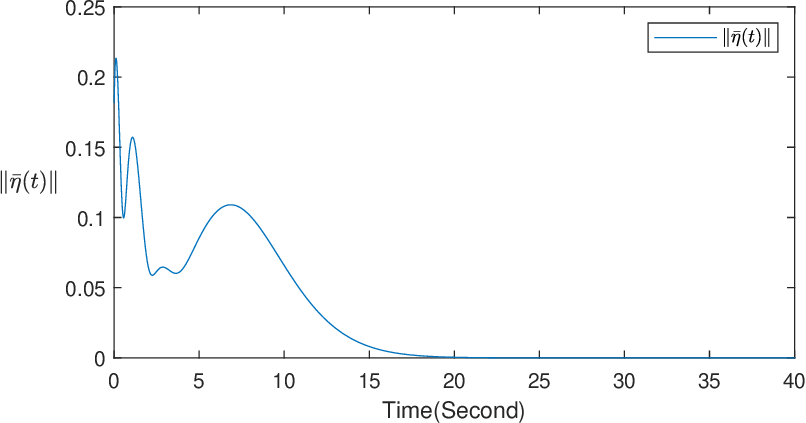,width=0.48\textwidth}
 \caption{Estimation error $\tilde{\eta}$ for Example 2}\label{eta-error}
 \end{figure} 
 \begin{figure}[ht]
 \epsfig{figure=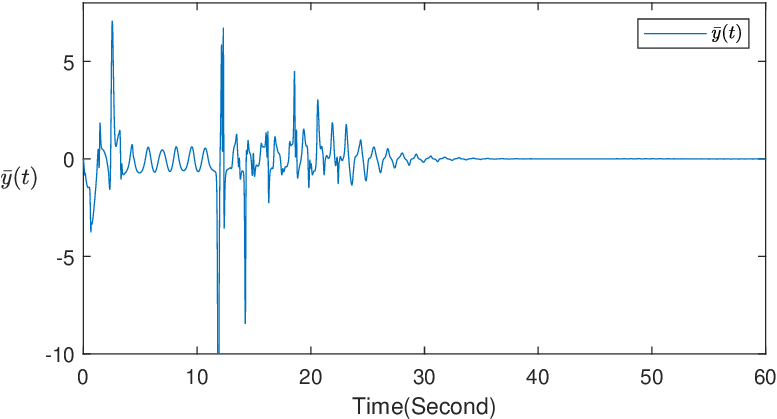,width=0.48\textwidth}
 \caption{Output estimation error for Example 2}\label{figpara-error}
\end{figure} 

\begin{figure}[ht]
 \epsfig{figure=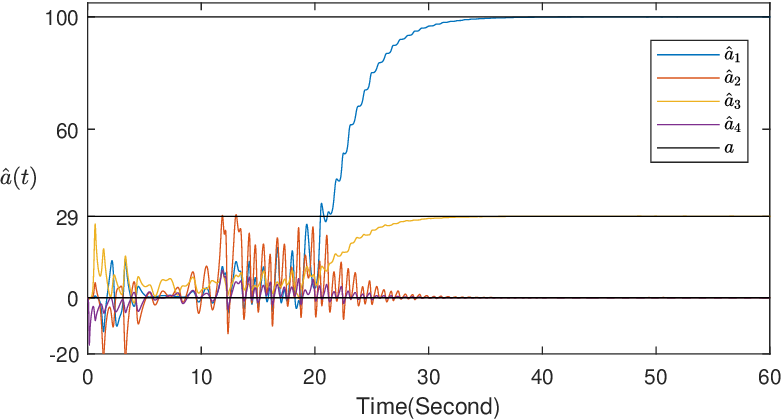,width=0.48\textwidth}
\caption{Trajectories of $\hat{a}(t)$ for Example 2}\label{figpara-true}
 \end{figure} 

\begin{figure}[ht]
 \epsfig{figure=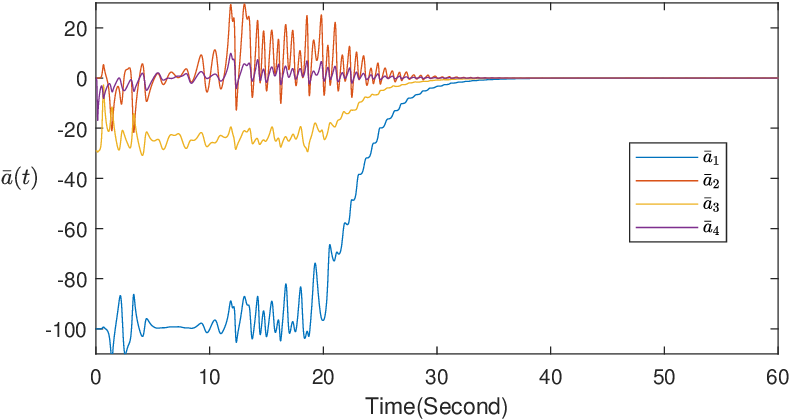,width=0.48\textwidth}
\caption{Parameter estimation error $\bar{a}(t)$ for Example 2}\label{figpara}
 \end{figure} 

\begin{figure}[ht]
\epsfig{figure=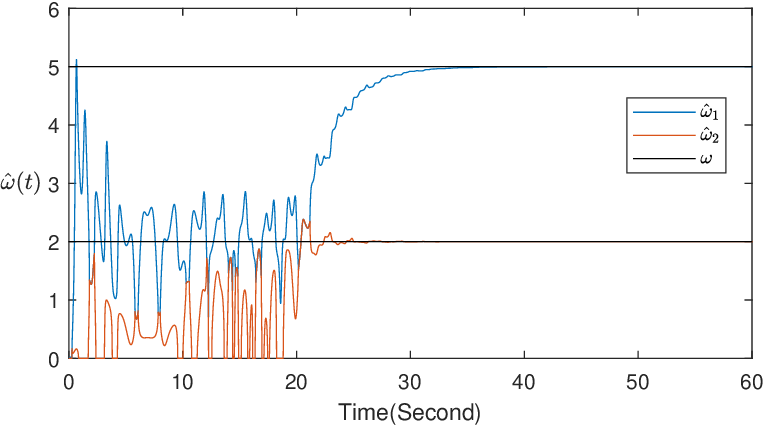,width=0.48\textwidth}
\caption{\myr Frequency estimation performance $\hat{\omega}(t)$ for Example 2}\label{figomegaPara}
\end{figure} 

 \begin{figure}[htbp]
 \centering
 \epsfig{figure=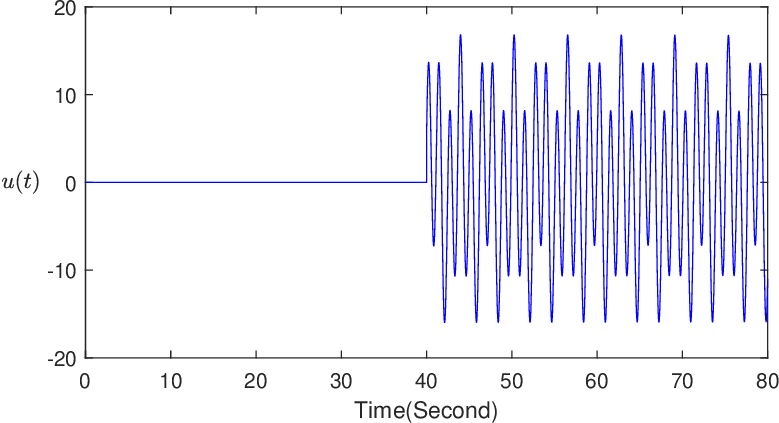,width=0.48\textwidth}  
  \caption{Control input of active automotive suspension system.}
 \end{figure} 

\begin{figure}[ht]
    \epsfig{figure=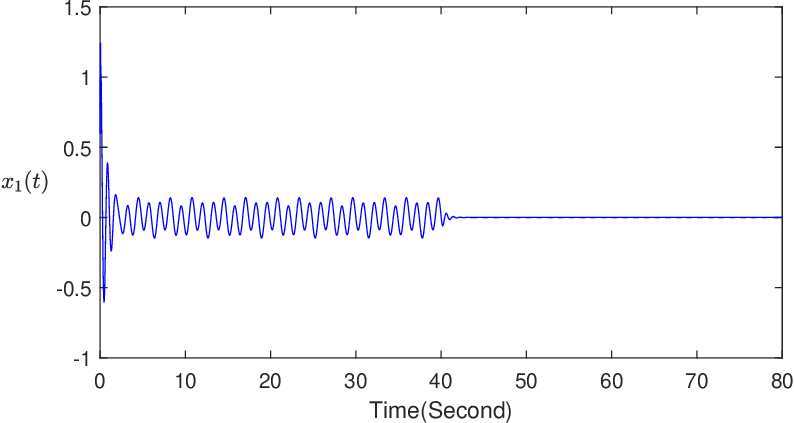,width=0.48\textwidth}  
  \caption{Trajectory of suspension deflection without and with active suspension.}\label{figessec}
\end{figure}

We borrow a modified example from \cite{yue1988alternative} and \cite{sferlazza2021state}  to illustrate our parameter estimation and active feedforward control design.
Consider the quarter-car active automotive suspension system described in Fig.~\ref{QCAASS}, which represents the automotive system at each wheel and is monitored by two separate sensors. The system has a spring $k_s$, a damper $b_s$, and an active force actuator $F_a$ as the input $u$. The sprung mass $m_s$ represents the quarter-car equivalent of the vehicle body mass. The unsprung mass $m_u$ represents the equivalent mass due to the axle and tire. The spring  $k_t$ represents the vertical stiffness of the tire. The variables
$z_s$, $z_u$ and $z_r$ are the vertical displacements from the static
equilibrium of the sprung mass, unsprung mass and the road, respectively.

The dynamics of this system are governed by
\begin{align}\label{qacau}
\dot{x}&=Ax+Bu+B_d\dot{z}_r,\notag\\
y_c&=x_{1},
\end{align}
where $u$ is the active force of the actuator, $\dot{z}_r$ is an input describing how the road profile enters into the system, and the state vector $x=\col(x_{1}, x_{2}, x_{3}, x_{4})$ includes the suspension deflection $x_{1}=z_s-z_u$, the absolute velocity $x_{2}=\dot{z}_s$ of the sprung mass $m_s$, the tire deflection $x_{3}=z_u-z_r$, and the absolute velocity $x_{4}=\dot{z}_u$ of the unsprung mass $m_u$. The matrices $A$, $B$, and $B_d$ are given by 
\begin{align*}A&=\left[
        \begin{array}{cccc}
          0 & 1 & 0 & -1 \\
          \frac{-k_s}{m_s} & \frac{-b_s}{m_s} & 0 & \frac{b_s}{m_s}\\
          0 & 0 & 0 & 1 \\
          \frac{k_s}{m_u} & \frac{b_s}{m_u}& \frac{-k_t}{m_u} & \frac{b_s+b_t}{-m_u} \\
        \end{array}
      \right],\\ 
                        B&=\col\big(0,\frac{1}{m_s}, 0, \frac{-1}{m_s}\big), \\
                        B_d&=\col\big(0, 0, -1, \frac{b_t}{m_u}\big). 
                        \end{align*}
where $m_s=2.40$, $m_u=0.36$, $b_s=9.8$, $k_s=160$, $k_t=1600$ and $b_t=0$. 
We assume that the road profile can be measured by the car and viewed as a disturbance described by \eqref{Exosys1} with \begin{align*}s(v,\sigma)&=\Phi (a)v,\\
\dot{z}_r&=\Gamma v\,{\myr \equiv}\,y_0,\end{align*}
where $v\in \mathds{R}^4$ is the state of the exosystem and has the characteristic polynomial $$\varsigma^4+a_4 \varsigma^2+a_3\varsigma^2+a_2\varsigma+a_1$$ with unknown true value vector $a=\col(100, 0, 29, 0)$ in \eqref{stagerator}. {\myr As in \cite{marino2002global}, the output of the exosystem is given by $$y_0= d_{1}\sin( \omega_1 t+\phi_{1})+d_{2}\sin(\omega_2 t+\phi_2)$$ where $\omega_1=2$ and $\omega_2=5$ are the unknown frequencies, $d_{1}$ and $d_2$ are the unknown amplitudes, and $\phi_1$ and $\phi_2$ are the unknown phases. The Laplace transform of the multiple sinusoidal functions is 
\begin{align*}\frac{d_1 s \sin(\phi_1)+d_1\omega_1 \cos(\phi_1)}{s^2+\omega_1^2}&+\frac{d_2 s \cos(\phi_2)-d_2\omega_2 \sin(\phi_2)}{s^2+\omega_2} \\
=&\frac{B(s)}{s^4+\underbrace{(\omega_1^2+\omega_2^2)}_{a_3}s^2+\underbrace{\omega_1^2\omega_2^2}_{a_1}}\end{align*}
where $B(s)$ is a third-order polynomial. As a result, the pair $(\omega_1,\omega_2)$ satisfies the relationship:
$$ \omega_{1}\;\textnormal{or}\;\omega_{2}=\sqrt{\frac{a_3 \pm \sqrt{a_3^2-4a_1}}{2}}.$$}
We choose $K_x=0$ in \eqref{feedforward-contol}, $\mu=10^7$, $k_2=8000$, $m_1=1$, $m_2=8$, $m_3=28$,
$m_4=56$, $m_5=70$, $m_6=56$, $m_7=28$ and $m_8=8$. 
 The simulation starts with the initial conditions: 
 $\eta(0)=0$ $v(0)=\col(0;7;0;-133)$, $\hat{\zeta}_{a}=0$, $\hat{a}=0$ and $x=\col(0.5962,6.8197,0.4243,0.7145)$. 


Fig.\ \ref{eta-error} shows the trajectory of estimation error  $\tilde{\eta} =\eta -{\eta}^{\star}$.
Fig.\ \ref{figpara-error} contains the output estimation error trajectories, $\bar{y} =\hat{y} -y_0$.
The parameter estimates $\hat{a}$ are shown in Fig.\  \ref{figpara-true}, with the corresponding parameter estimation error, $\bar{a} =\hat{a} -a$, given in Fig.\  \ref{figpara}.
The frequency estimates $\hat{\omega}$ are shown in Fig.\ \ref{figomegaPara}. 
Finally, Fig.\ \ref{figessec} shows the trajectory of suspension deflection without or with active suspension. The control signal, $u(t)$, is active for $t\geq 40$. 


\section{Conclusion}\label{Sec-Conclusion}%
{ The global robust output regulation problem for a class of nonlinear systems with output feedback is solved using the proposed nonparametric learning solution framework for a linear generic internal model design. 
The approach relaxes the linearity assumption from \cite{xu2019generic} to handle a steady-state generator that is polynomial in the exogenous signal. In doing so, we demonstrate that the persistency of excitation remains a necessary condition in the generic internal design by using the results in \cite{liu2009parameter}.
The nonparametric learning framework can solve the linear time-varying equation online, ensuring that a suitable nonlinear continuous mapping always exists.
With the help of the proposed framework, the robust nonlinear output regulation problem is converted into a robust non-adaptive stabilization problem for the augmented system with integral Input-to-State Stable (iISS) inverse dynamics. 
%
%
We also apply the nonparametric learning framework to globally estimate the unknown frequencies, amplitudes, and phases of the $n$ sinusoidal
components with bias independent of using adaptive techniques that differ from traditional adaptive control methods used in available approaches \cite{marino2002global,xia2002global,hou2011parameter,pin2019identification}. 
The results also overcome the global frequency estimation requirements \cite{praly2006new}. This extends the frequency estimation to the sum of a bias and $n$ sinusoidal and removes the assumption that the actual frequency lies within a known compact set.
We have also shown that the explicit nonlinear continuous mapping can directly provide the estimated parameters after a specific time and will exponentially converge to the unknown frequencies.
The nonparametric learning framework naturally avoids the need for the construction of a Lyapunov function or an input-to-state stability (ISS) Lyapunov function for the closed-loop system with positive semidefinite derivatives.
Finally, based on this non-adaptive observer, we synthesize a feedforward control design to solve the output regulation using the proposed nonparametric learning framework using a certainty equivalence principle.

}


\section{Appendix}\label{Appendix}%
\subsection{The Proof of Properties \ref{property1} and \ref{property2}}\label{AppendixA}%
\begin{proof}%
We first show that Property \ref{property1} is satisfied.
Under Assumption \ref{H2}, there exists a continuous function $V_{\bar{z}}(\bar{z})$ satisfying 
$$\underline{\alpha}_{\bar{z}}(\|\bar{z} \|)\leq V_{\bar{z} }(\bar{z}) \leq \overline{\alpha}_{\bar{z}}(\|\bar{z} \|)$$
for some class $\mathcal{K}_{\infty}$ functions $\underline{\alpha}_{\bar{z}}(\cdot)$ and $\overline{\alpha}_{\bar{z}}(\cdot)$ such that, for any $v\in \mathds{V}$, along the trajectories of the $\bar{z} $ subsystem,
$$\dot{V}_{\bar{z} }\leq-\alpha_{\bar{z}}(\|\bar{z} \|)+ \gamma (e),$$
where $\alpha_{\bar{z}}(\cdot)$ is some known class $\mathcal{K}_{\infty}$ function satisfying $\limsup\limits_{\varsigma\rightarrow 0^{+}}(\alpha_{\bar{z}}^{-1}(\varsigma^2)/\varsigma)< + \infty$, and $\gamma (\cdot)$ is some known smooth positive definite function.

Define the Lyapunov function candidate $V_{\bar{\eta}}(\bar{\eta})=\bar{\eta}^\top P\bar{\eta}$, where $P$ is the unique positive definite symmetric matrix that satisfies $PM+M^\top P=-2I$. Let $\lambda_p$ and $\lambda_P$ be the minimum and maximum eigenvalues of $P$.
The time
derivative of $V_{\bar{\eta}}(\bar{\eta})$ along the trajectories of \eqref{augmen-1b} are
\begin{align}
\dot{V}_{\bar{\eta}}\leq&-\|\bar{\eta}\|^2+\|P\bar{p}(\bar{z} ,e ,\mu)\|^2.\nonumber
\end{align}
Since $\bar{p}\!\left(0 ,0 ,v\right)=0$, for all $v\in \mathds{V}$, by Lemma 7.8 in \cite{huang2004nonlinear},
$$\|P\bar{p} \big(\bar{z}, e ,v\big)\|^2\leq\pi_1 (\bar{z})\|\bar{z} \|^2+\phi_1(e )e^2$$
for some known smooth functions $\pi_1(\cdot)\geq 1$ and $\phi_1(\cdot)\geq 1$. Then, we have 
\begin{align}
\dot{V}_{\bar{\eta}}\leq&-\|\bar{\eta}\|^2+\pi_1 (\bar{z} )\|\bar{z} \|^2+\phi_1 (e )e^2.\nonumber
\end{align}

Define the Lyapunov function $U_{\bar{z}}(\bar{z})=\int_{0}^{V_{\bar{z}}(\bar{z})}\kappa(s)ds$, where the positive function $\kappa(\cdot)$ will be specified later. The time
derivative of $U_{\bar{z}}(\bar{z})$ along \eqref{augmen-1b} can be evaluated as
\begin{align}
    \dot{U}_{\bar{z}}\leq -\kappa\circ V_{\bar{z}}(\bar{z}) \left[\alpha_{\bar{z}}(\|\bar{z} \|)+ \gamma (e)\right].\nonumber
\end{align}
By the changing supply rate technique \cite{sontag1995changing}, we will have
\begin{align}\label{UtildeZSN}
    \dot{U}_{\bar{z}} \leq& -\frac{1}{2}\kappa\circ \underline{\alpha}_{\bar{z}}(\|\bar{z} \|) \alpha_{\bar{z}}(\|\bar{z} \|)+\kappa \circ\theta(e) \gamma (e),
\end{align}
where {\myr $\theta(\cdot)$ is a smooth positive definite function} defined as $\theta\,{\myr \equiv}\,\bar{\alpha}_{\bar{z}} \circ \alpha_{\bar{z}}^{-1}\circ (2 \gamma)$. It is noted that $\alpha_{\bar{z}}(\cdot)$ is some known class $\mathcal{K}_{\infty}$ function satisfying $\limsup\limits_{\varsigma\rightarrow 0^{+}}\,(\alpha_{\bar{z}}^{-1}(\varsigma^2)/\varsigma)< + \infty$. Then, there exists a smooth function $\alpha_0(\|\bar{z} \|)$ such that 
$$\alpha_0(\|\bar{z} \|)\alpha_{\bar{z}}(\|\bar{z} \|)\geq \|\bar{z} \|^2.$$
As $\limsup\limits_{\varsigma\rightarrow 0^{+}}(\alpha_{\bar{z}}^{-1}(\varsigma^2)/\varsigma)< + \infty$, there exists constant $l_1\geq1$ such that $\alpha_{\bar{z}}(\|\bar{z} \|) \geq \|\bar{z} \|^2/l_1^2$ for all $ \|\bar{z} \| \leq 1$. Besides, $\alpha_{\bar{z}}(\|\bar{z} \|)$ is of class $\mathcal{K}_{\infty}$, there exists a constant $l_2> 0$ such that $\alpha_{\bar{z}}(\|\bar{z} \|)\geq l_2$ for all $\varsigma\geq 1$. Hence, we will have $$\alpha_0(\|\bar{z} \|)\alpha_{\bar{z}}(\|\bar{z} \|)\geq \|\bar{z} \|^2$$ for any $\alpha_0(\|\bar{z} \|)\geq l_1^2+l_2\|\bar{z} \|^2$. We can choose a positive smooth non-decreasing function $\kappa(\cdot)$ such that $$\frac{1}{2}\kappa\circ \underline{\alpha}_{\bar{z}}(\|\bar{z} \|) \geq \alpha_0(\|\bar{z} \|)
\times(\pi_1 (\bar{z} )+1).$$
Let us consider the $\bar{Z}$-subsystem of \eqref{augmen-1a} and \eqref{augmen-1b}. 
Define the Lyapunov function candidate $V_0(\bar{Z})=U_{\bar{z}}(\bar{z})+V_{\bar{\eta}}(\bar{\eta})$.
The time
derivative of $V_0(\bar{Z})$ along the trajectories of \eqref{augmen-1a} and \eqref{augmen-1b} are given by
\begin{align}\label{zdlya}
    \dot{V}_0\leq&-\frac{1}{2}\kappa\circ {\myr \underline{\alpha}_{\bar{z}}(\|\bar{z} \|) \alpha_{\bar{z}}(\|\bar{z} \|)}+\kappa \circ\theta(e) \gamma (e)\nonumber\\
    &-\|\bar{\eta}\|^2+\pi_1 (\bar{z})\|\bar{z} \|^2+\phi_1 (e )e^2\nonumber\\
    \leq& -\|\bar{Z} \|^2+ \bar{\gamma}(e),
\end{align}
where $\bar{\gamma}(e)=\kappa \circ\theta(e) \gamma(e)+\phi_1 (e )e^2$. Using Lemma 11.3 in \cite{chen2015stabilization}, we can choose class $\mathcal{K}_{\infty}$ functions $\underline{\alpha}_{0}\left(\cdot\right)$ and $\overline{\alpha}_{0}\left(\cdot\right)$ such that
\begin{align*}\underline{\alpha}_{0}\left(\|\bar{Z} \|\right)\leq \int_{0}^{\underline{\alpha}_{\bar{z}}\left(\|\bar{z} \|\right)}\kappa(\tau)d\tau+\lambda_p\|\bar{\eta}\|^2,\\
\overline{\alpha}_{0}\left(\|\bar{Z} \|\right) \geq \int_{0}^{\overline{\alpha}_{\bar{z}}\left(\|\bar{z} \|\right)}\kappa(\tau)d\tau+\lambda_P\|\bar{\eta}\|^2.
\end{align*}
We now verify Property \ref{property2}. 
It can be verified that the function $\bar{g}(\bar{z}, e, \bar{\eta}, \mu)$ is smooth and vanishes at $\col(\bar{z}, e, \bar{\eta})=\col(0,0,0)$. By \cite[Lemma 11.1]{chen2015stabilization}, there exist positive smooth functions $\gamma_{g0}(\cdot)$ and $\gamma_{g1}(\cdot)$ such that
\begin{align*}\|\bar{g}(\bar{Z},e,\mu)\|^2\leq& \gamma_{g0}(\bar{Z})\|\bar{Z}\|^2+e^2\gamma_{g1}(e)
\end{align*}
for $\mu\in \mathds{V}\times \mathds{W}\times \mathds{S}$. This completes the proof.
\end{proof}

\subsection{Proof of Theorem \ref{Theorem-1}}\label{AppendixC}%
\begin{proof}%
The error dynamics \eqref{augmen-1} with control \eqref{ESC-1} are given by
\begin{equation}\label{av1Esc-1}\begin{aligned}
 \dot{\bar{z}} &= \bar{f}\!\left(\bar{z},e,\mu\right), \\
    \dot{\bar{\eta}}&= M\bar{\eta}+\bar{p}(\bar{z},e,\mu),\\
\dot{e}&=\bar{g}\left(\bar{z},\bar{\eta},e,\mu\right)+ \bar{b}(\bar{z},e,\mu)\left(\bar{\chi}(\bar{\eta},e,\mu) -k e{\rho}(e)\right),
\end{aligned}\end{equation}
where $\bar{\chi}(\bar{\eta},e,\mu)= \chi\!\left(\bm{\eta}^{\star}+\bar{\eta}+\bar{b}(\bar{z},e,\mu)^{-1}Ne\right) -\chi(\bm{\eta}^{\star})$. 
From \cite{KE2003} or \cite{marconi2008uniform}, we have that $\chi(\cdot)$ is a {\myr continuously differentiable function.} 
 It can also be verified that the function {\myr $\bar{b}(\bar{z},e,\mu)\bar{\chi}(\bar{\eta},e,\mu)$} is continuous and vanishes at $\col(\bar{z}, e, \bar{\eta})=\col(0,0,0)$.
 By \cite[Lemma 11.1]{chen2015stabilization}, there exist positive smooth functions $\gamma_{1\chi}(\cdot)$ and $\gamma_{2\chi}(\cdot)$ such that
\begin{align*}{\myr |\bar{b}(\bar{z},e,\mu) \bar{\chi}(\bar{\eta},e,\mu)|^2}\leq&  \gamma_{1\chi}(\bar{Z})\|\bar{Z}\|^2+ e^2\gamma_{2\chi}(e)
\end{align*}
for $\mu\in \mathds{V}\times \mathds{W}\times \mathds{S}$.
From Property \ref{property1}, with the same development and by using the changing supply rate technique \cite{sontag1995changing} again, given any smooth function $\Theta_z (\bar{Z} )>0$, there exists a continuous function $V_{2}(\bar{Z})$ satisfying
$$\underline{\alpha}_{2}\big(\big\|\bar{Z} \big\|^2\big)\leq V_{2}\big( \bar{Z}  \big)\leq\overline{\alpha}_{2}\big(\big\|\bar{Z} \big\|^2\big)$$
 for some class $\mathcal{K}_{\infty}$ functions $\underline{\alpha}_{2}(\cdot)$ and $\overline{\alpha}_{2}(\cdot)$, such that, for all $\mu\in \Sigma$, along the trajectories of the $Z $ subsystem, $$\dot{V}_{2} \leq-\Theta_z (\bar{Z} )\big\|\bar{Z} \big\|^2+ \hat{\gamma} \left(e\right)e^2, $$
where $\hat{\gamma} (\cdot)\geq 1$ is some known smooth positive function.
Next, we consider the augmented system \eqref{av1Esc-1}  
and pose the following Lyapunov function candidate:
\begin{align}
V(\bar{Z},e)=V_2(\bar{Z})+e^2.\nonumber
\end{align}
It follows that $V(\bar{Z},e)$ is globally positive definite and radially unbounded. From Property \ref{property2}, we have that the derivative of $V$ along the trajectories of the system \eqref{av1Esc-1} subject  to the controller \eqref{ESC-1} satisfies
\begin{align}
\dot{V}&=\dot{V}_2+e^\top\dot{e}\nonumber\\
&\leq\dot{V}_2+e^2+\gamma_{g0}(\bar{Z})\|\bar{Z}\|^2+e^2\gamma_{g1}(e)-2k\bar{b}(\bar{z},e,\mu)e^2\rho(e)\nonumber\\
&\leq-(\Theta_z(\bar{Z})-\gamma_{g0}(\bar{Z})-\gamma_{1\chi}(\bar{Z}))\big\|\bar{Z}\big\|^2\nonumber\\
&\quad -\left(2b(\mu)k\rho(e)-\hat{\gamma} \left(e\right)-1-\gamma_{g1}(e)-\gamma_{2\chi}(e)\right)e ^2.
 \end{align}
It is noted that $b(\bar{z},e,\mu)$ is lower bounded by some positive constant $b^*$. Let the smooth functions $\Theta_z(\cdot)$, ${\rho}(\cdot)$, and the positive number ${k}$ be such that $\Theta_z (\bar{Z})\geq \gamma_{g0} (\bar{Z} )+\gamma_{1\chi}(\bar{Z})+1$ and $\rho(e)\geq \{\hat{\gamma} \left(e \right), \gamma_{2\chi}(e),\gamma_{g1}(e), 1\}$, and $k \geq \frac{4}{2b^*}\equiv k^*$. As a result, we obtain
$$\dot{V}\leq -\big\|\bar{Z}\big\|^2-e^2.$$
It can then be concluded that system \eqref{av1Esc-1} is globally uniformly asymptotically stable for all $\mu \in \mathds{V}\times \mathds{W}\times \mathds{S}$.
 This completes the proof.

\end{proof}

\bibliographystyle{ieeetr}
\bibliography{myref}

\begin{IEEEbiography}[{\includegraphics[width=1in,height=1.25in, clip,keepaspectratio]{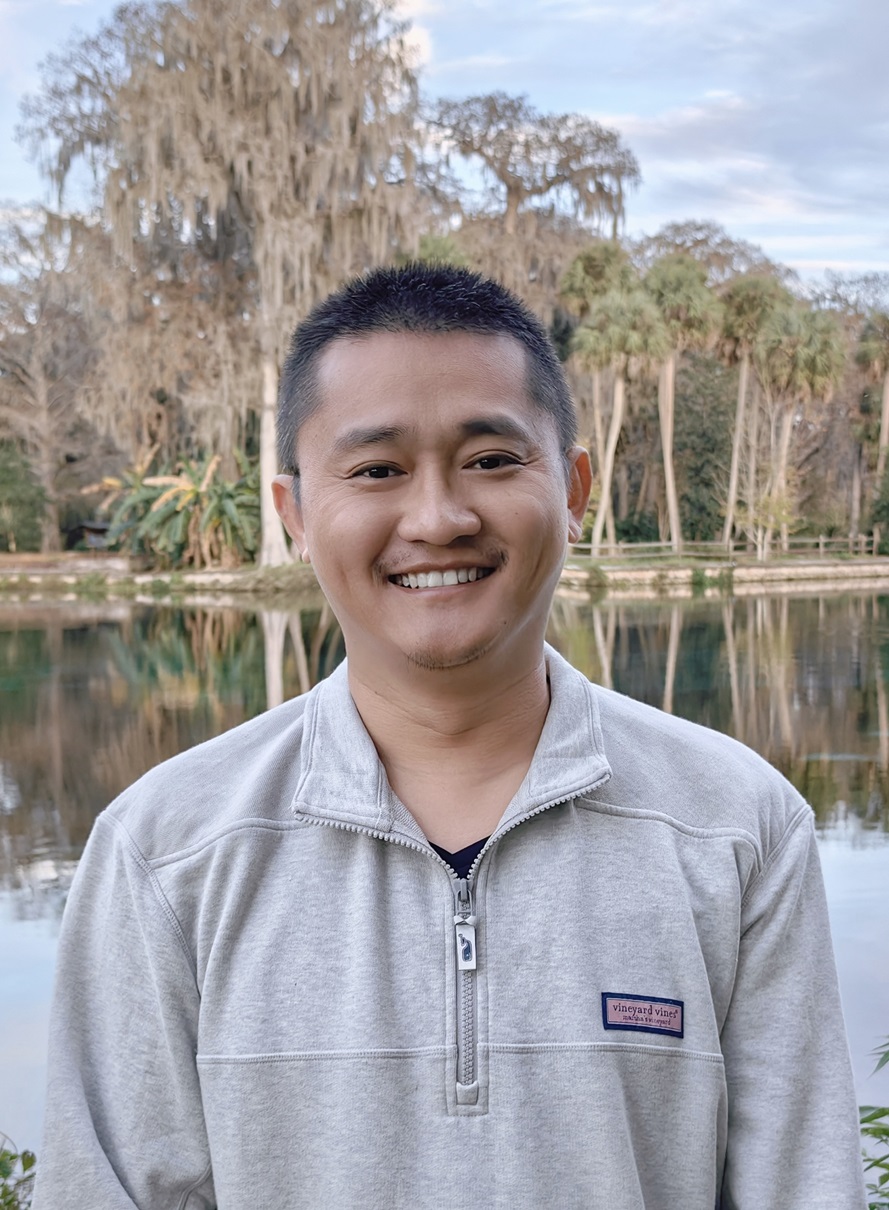}}]{Shimin Wang} 
received the B.Sci. in Mathematics and Applied Mathematics and M.Eng. in Control Science and Control Engineering from Harbin Engineering University in 2011 and 2014, respectively. He then received the Ph.D. in Mechanical and Automation Engineering from The Chinese University of Hong Kong in 2019. 

He was a recipient of the NSERC Postdoctoral Fellowship award in 2022. From 2014 to 2015, he
was an assistant engineer at the Jiangsu Automation
Research Institute, China State Shipbuilding Corporation Limited. From 2019
to 2023, he held postdoctoral positions at the University of Alberta and Queens University. He is
now a postdoctoral associate
at the Massachusetts Institute of Technology.
\end{IEEEbiography}

\begin{IEEEbiography}[{\includegraphics[height=1.25in, clip,keepaspectratio]{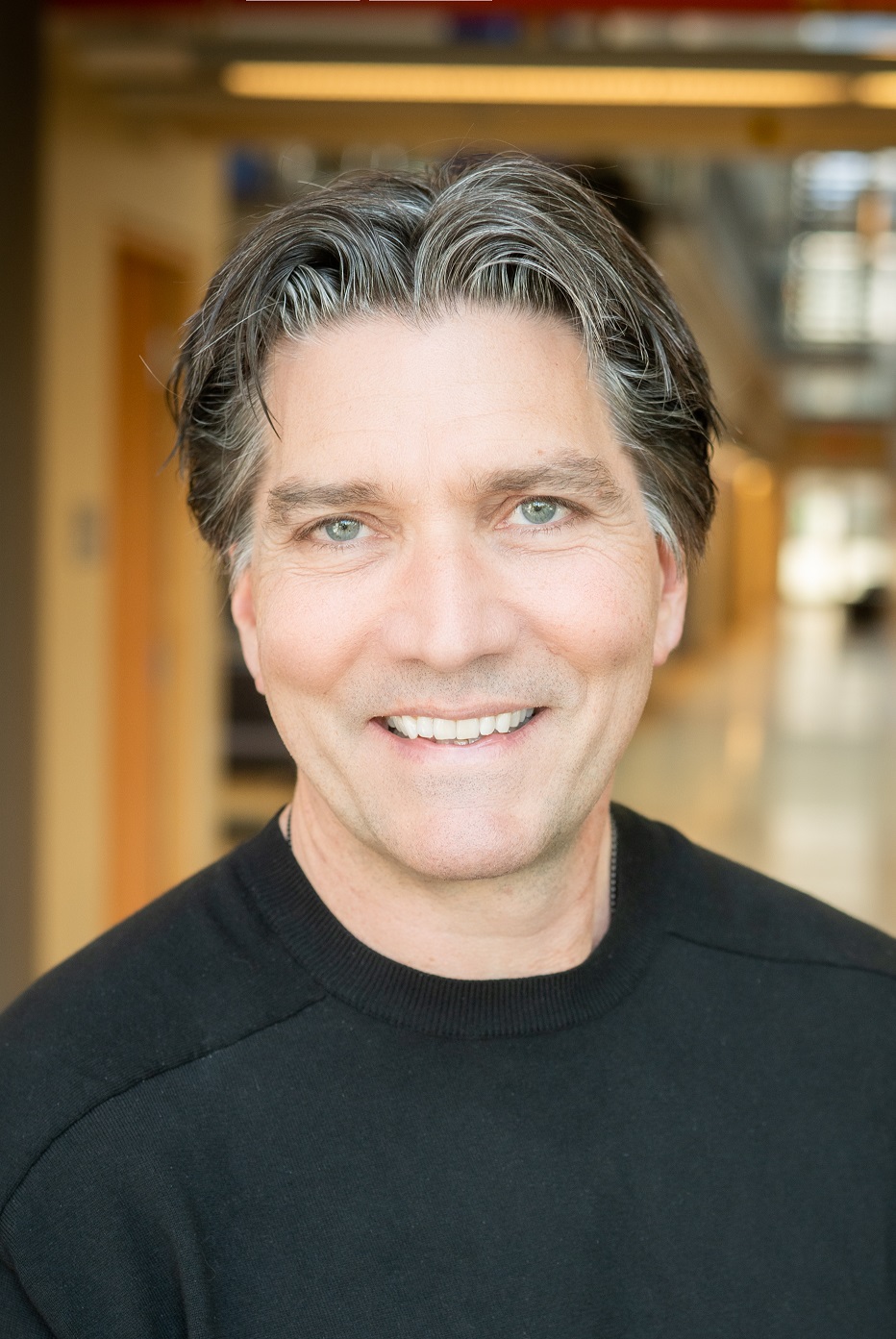}}]{Martin Guay} received a Ph.D. from Queen’s University, Kingston, ON, Canada in 1996. He is currently a Professor in the Department of Chemical Engineering at Queen’s University. His current research interests include nonlinear control systems, especially extremum-seeking control, nonlinear model predictive control, adaptive estimation and control, and geometric control. 

He was a recipient of the Syncrude Innovation Award, the D. G. Fisher from the Canadian Society of Chemical Engineers, and the Premier Research Excellence Award. He is a Senior Editor of IEEE Control Systems Letters. He is the Editor-in-Chief of the Journal of Process Control. He is also an Associate Editor for Automatica, IEEE Transactions on Automatic Control and the Canadian Journal of Chemical Engineering.
  \end{IEEEbiography}

\begin{IEEEbiography}[{\includegraphics[width=1in,height=1.25in, clip,keepaspectratio]{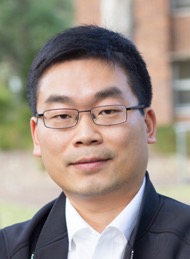}}]{Zhiyong Chen} received the B.E. in Automation from the University of Science and Technology of China, Hefei, China, in 2000, and the M.Phil. and Ph.D. in Mechanical and Automation Engineering from the Chinese University of Hong
Kong, in 2002 and 2005, respectively. He worked as a Research Associate at the University of Virginia, Charlottesville, VA, USA, from 2005 to 2006. In 2006, he joined the University of Newcastle, Callaghan, NSW, Australia, where he is currently a Professor. 

He was also a Changjiang Chair Professor at Central South University, Changsha, China. His research interests include
nonlinear systems and control, biological systems, and reinforcement learning.
He is/was an Associate Editor of Automatica, IEEE Transactions on
Automatic Control, IEEE Transactions on Neural Networks and Learning
Systems, and IEEE Transactions on Cybernetics.
\end{IEEEbiography}

\begin{IEEEbiography}[{\includegraphics[width=1in,height=1.25in, clip,keepaspectratio]{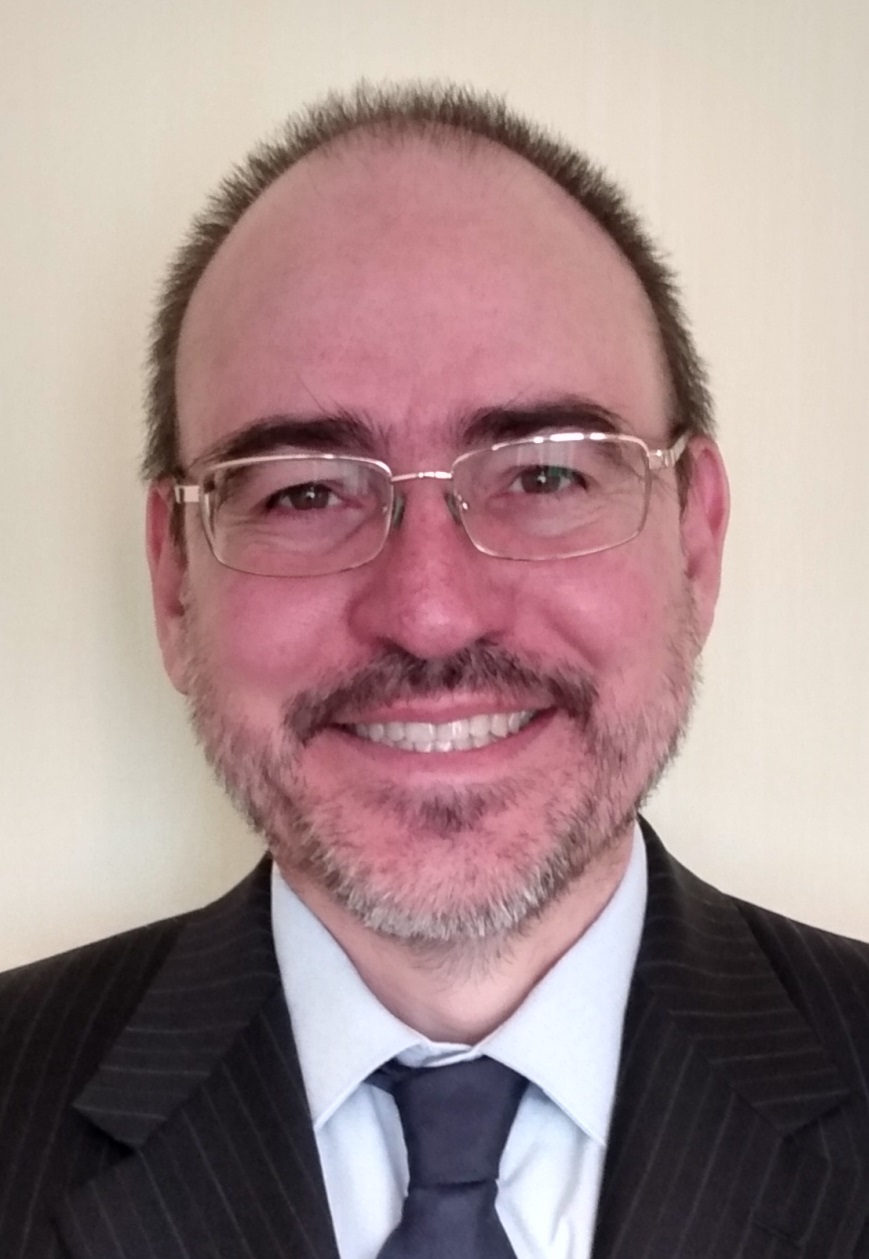}}]{Richard D. Braatz} is the Edwin R. Gilliland Professor at the Massachusetts Institute of Technology (MIT) where he does research in applied mathematics and robust optimal control theory and their application to advanced manufacturing systems. He received an M.S. and Ph.D. from the California Institute of Technology and was on the faculty at the University of Illinois at Urbana–Champaign and was a Visiting Scholar at Harvard University before moving to MIT. He is a past Editor-in-Chief of IEEE Control Systems Magazine and a past President of the American Automatic Control Council, and is a Vice President of IFAC. Honors include the AACC Donald P. Eckman Award, the Curtis W. McGraw Research Award from the Engineering Research Council, the Antonio Ruberti Young Researcher Prize, and best paper awards from IEEE- and IFAC-sponsored control journals. He is a member of the U.S. National Academy of Engineering and a Fellow of IEEE and IFAC.
\end{IEEEbiography}

\end{document}